\documentclass[ reprint, superscriptaddress, nofootinbib, amsmath,amssymb, aps ]{revtex4-2}

\usepackage{lineno}

\usepackage{graphicx}
\usepackage{dcolumn}
\usepackage{bm}
\usepackage[colorlinks,linkcolor=blue,urlcolor=blue, citecolor=blue]{hyperref}
\usepackage[whole]{bxcjkjatype}

\usepackage{dsfont}
\usepackage{physics}
\usepackage{nicefrac}
\usepackage{enumitem}
\usepackage{tcolorbox}
\usepackage{needspace}
\usepackage{mathtools}
\usepackage{tabularx}
\newcolumntype{C}{>{\centering\arraybackslash}X}

\usepackage{amsmath, amssymb, amsfonts, amsthm}
\usepackage{thmtools, thm-restate}

\DeclareMathOperator{\spn}{span}
\DeclareMathOperator{\Sym}{Sym}

\newtheorem{theorem}{Theorem}
\newtheorem*{theorem*}{Theorem}

\newtheorem{corollary}[theorem]{Corollary}
\newtheorem{lemma}[theorem]{Lemma}
\newtheorem{definition}[theorem]{Definition}
\newtheorem{proposition}[theorem]{Proposition}

\newcommand{\red}{\color{red}}

\allowdisplaybreaks

\begin{document}

\title{Generalized quantum Stein's lemma: Redeeming second law of resource theories}

\author{Hayata Yamasaki}
\email{hayata.yamasaki@gmail.com}
\affiliation{Department of Physics, Graduate School of Science, The Univerisity of Tokyo, 7-3-1 Hongo, Bunkyo-ku, Tokyo, 113-0033, Japan}
\author{Kohdai Kuroiwa}
\email{kkuroiwa@uwaterloo.ca}
\affiliation{Institute for Quantum Computing and Department of Combinatorics and Optimization, University of Waterloo, Ontario, Canada, N2L 3G1}
\affiliation{Perimeter Institute for Theoretical Physics, Ontario, Canada, N2L 2Y5}

\begin{abstract}
\textbf{{\red [Note: After the first version of this manuscript was uploaded, the authors of Ref.~\cite{berta2023gap} pointed out an issue about a part of the claims in the previous version of Refs.~\cite{10206734,10129917} used in our analysis. Due to this issue, the analysis in the previous version of this manuscript can no longer be considered complete proof of the generalized quantum Stein's lemma. This version is a temporal update to add this note. We are planning to update the manuscript further to explain the issue and what conditions we will additionally need to complete the proof of the generalized quantum Stein's lemma.]}}

\vspace{1cm}

The second law lies at the heart of thermodynamics, characterizing the convertibility of thermodynamic states by a single quantity, the entropy. A fundamental question in quantum information theory is whether one can formulate an analogous second law characterizing the convertibility of resources for quantum information processing. In 2008, a promising formulation was proposed, where quantum-resource convertibility is characterized by the optimal performance of a variant of another fundamental task in quantum information processing, quantum hypothesis testing. The core of this formulation was to prove a lemma that identifies a quantity indicating the optimal performance of this task---the generalized quantum Stein's lemma---to seek out a counterpart of the thermodynamic entropy in quantum information processing. However, in 2023, a logical gap was found in the existing proof of the generalized quantum Stein's lemma, throwing into question once again whether such a formulation is possible at all. In this work, we construct a proof of the generalized quantum Stein's lemma by developing alternative techniques to circumvent the logical gap of the existing analysis. With our proof, we redeem the formulation of quantum resource theories equipped with the second law as desired. These results affirmatively settle the fundamental question about the possibility of bridging the analogy between thermodynamics and quantum information theory.
\end{abstract}

\maketitle

\paragraph*{Introduction.}
The advantages of quantum information processing over conventional classical information processing arise from the intrinsic properties of quantum mechanics, such as entanglement and coherence~\cite{N4}.
Quantum resource theories (QRTs) emerge as an operational framework for studying the manipulation and quantification of these various types of quantum properties as resources for quantum information processing~\cite{Chitambar2018,Kuroiwa2020}.
A QRT is formulated by considering a restricted class of quantum operations called free operations.
Quantum states that can be obtained from any initial state by free operations are called free states.
Non-free states are considered to be resources to overcome the restriction on operations.

The establishment of such an operational framework has been a traditional and remarkably successful approach in physics.
Ever since the invention of the steam engine ignited the Industrial Revolution, energy resources have driven human evolution, and technological innovations have improved the efficiency of energy conversion for useful applications; however, physicists have known a universal theory that never changes no matter how further the technologies may advance---thermodynamics~\cite{carnot1890reflections,doi:10.1080/14786445608642141,Thomson_1853,LIEB19991,Lieb2004,10.1063/1.883034}.
Thermodynamics reveals fundamental limits on the efficiency of using energy resources under any possible operations within a few axioms that prohibit unphysical state conversions.
It is applicable not only to the steam engine as originally motivated but, for example, to the feasibility of chemical reactions~\cite{lewis1923thermodynamics,guggenheim1933modern} and the energy consumption of computation~\cite{5392446,meier2023energyconsumption}.
The universal applicability of thermodynamics stems from its axiomatic formulation, which holds regardless of the details of the physical realization of the operations.
At the heart of thermodynamics is the existence of a single real-valued function of the state, the entropy, which quantitatively indicates whether a state can be converted into another state under any possible adiabatic operations, as postulated by the \textit{second law of thermodynamics}~\cite{LIEB19991,Lieb2004,10.1063/1.883034}.
In the same way, the key goal of QRTs is to establish a universal framework to clarify quantitative understandings and fundamental limitations of the manipulation of quantum resources.

\begin{figure}
    \centering
    \includegraphics[width=3.4in]{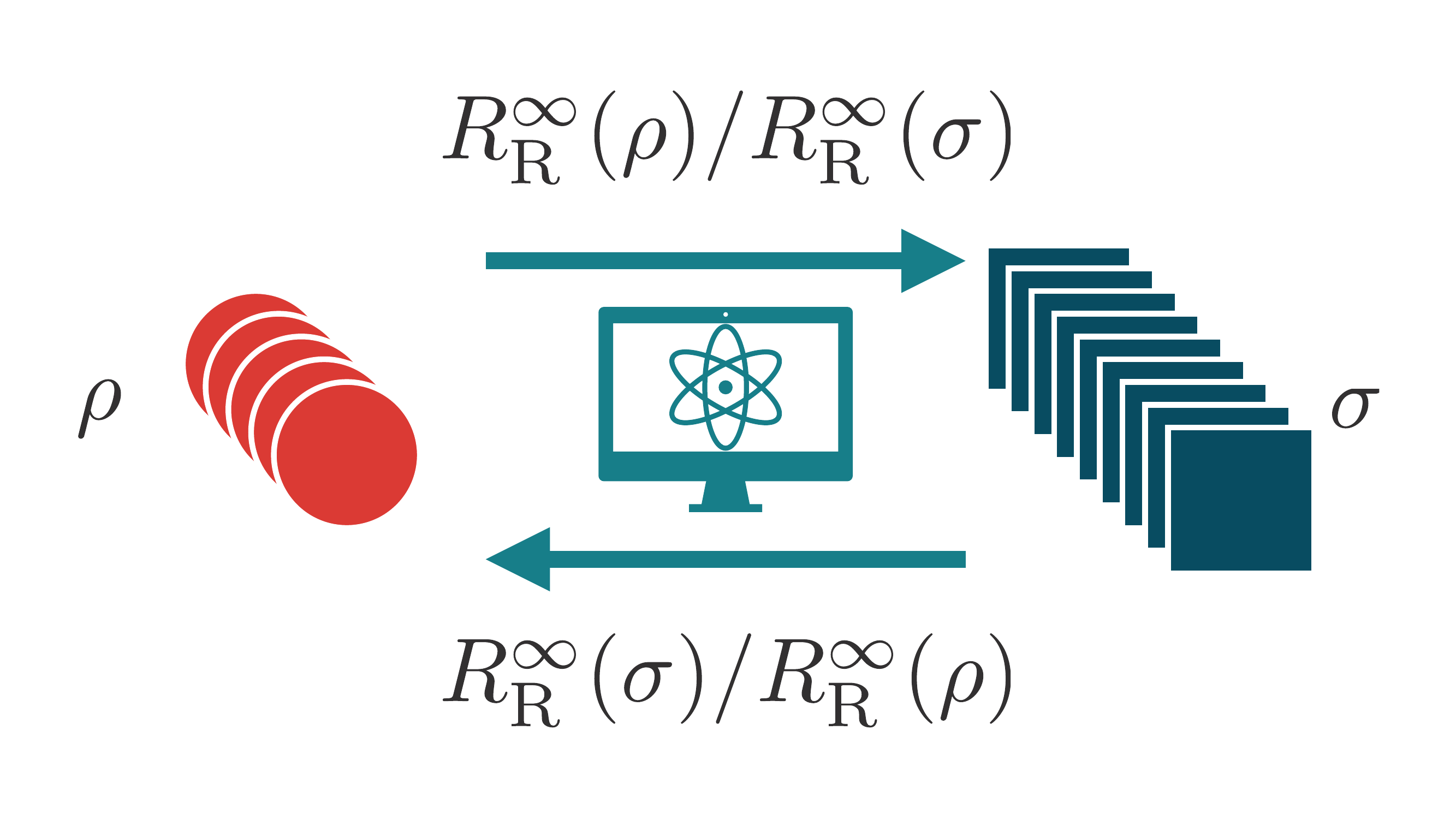}
    \caption{The second law of quantum resource theories (QRTs). Our main result, the proof of the generalized quantum Stein's lemma, leads to an axiomatic formulation of QRTs equipped with the second law in analogy with the second law of thermodynamics, where a single function, i.e., the regularized relative entropy of resource $R_\mathrm{R}^\infty$, characterizes the necessary and sufficient condition on the asymptotic convertibility between many copies of quantum resources, i.e., states $\rho$ (red circles) and $\sigma$ (blue squares), at the optimal rate as shown in the figure, as with the entropy in thermodynamics.}
    \label{fig:second_law}
\end{figure}

A fundamental task in QRTs is an asymptotic conversion of quantum states~\cite{Chitambar2018,Kuroiwa2020}.
In a spirit of information theory~\cite{6773024,cover2012elements},
the goal of this task is to convert many independently and identically distributed (IID) copies of state $\rho$ into as many copies of state $\sigma$ as possible, within a vanishing error, under a restricted class of operations considered for manipulating resources.
The maximum number $r(\rho\to\sigma)$ of copies of $\sigma$ obtained per $\rho$ is called the rate of the asymptotic conversion.
In view of the success of thermodynamics, it is also crucial to seek a universal formulation of QRTs equipped with an analogous \textit{second law of QRTs} (Fig.~\ref{fig:second_law}), i.e., with a single function of the quantum state offering a necessary and sufficient condition on the feasibility of asymptotic conversion at rate $r(\rho\to\sigma)$.
Despite thermodynamics and QRTs both taking operational approaches, the establishment of the analogous second law of QRTs is indeed challenging in general.
For example, a conventional way of introducing free operations in the theory of entanglement is based on a bottom-up approach, where a class of operations is defined by specifying what we can realize in the laboratories; in the bottom-up approach, local operations and classical communication (LOCC) are one of the conventional choices of free operations~\cite{Horodecki2009,chitambar2014everything}.
But indeed, it is known that the second law does not hold for entanglement theory under LOCC~\cite{PhysRevLett.86.5803}; as suggested by this example, it is by no means straightforward to formulate QRTs with the second law as long as one takes this bottom-up approach.

By contrast, as in the second law of thermodynamics formulated for the axiomatically defined class of adiabatic operations, it should also be promising to pursue a more universal formulation of QRTs with the second law from the first principles based on the axiomatic approach, where the restriction on the operations is specified by what cannot be done rather than what can be done.
In 2008, Refs.~\cite{Brand_o_2008,brandao2010reversible,Brandao2010} made substantial progress toward establishing such a second law in entanglement theory, which was later extended to more general QRTs in Ref.~\cite{Brandao2015}.
Their crucial observation was that the second law of QRTs may hold under an axiomatically defined class of quantum operations that should not generate resources from free states up to an asymptotically vanishing amount. 
Under this class of operations, a nontrivial bridge was built between the asymptotic conversion of quantum states and a variant of another fundamental task in quantum information theory, quantum hypothesis testing~\cite{hiai1991proper,887855}.
The goal of this task is to distinguish $N$ IID copies of quantum state $\rho^{\otimes N}$ from any state $\sigma_N$ in the set of free states that may be in a non-IID form.
A substantial challenge arises from the non-IIDness of $\sigma_N$, but Ref.~\cite{Brandao2010} tried to develop a technique for addressing this non-IIDness, to establish a lemma characterizing the optimal performance of this task, which is called the generalized quantum Stein's lemma.
If the generalized quantum Stein's lemma holds true, then QRTs equipped with the second law can be formulated~\cite{Brand_o_2008,brandao2010reversible,Brandao2010,Brandao2015}.
However, in 2023, a logical gap was found in the original analysis of the generalized quantum Stein's lemma in Ref.~\cite{Brandao2010}, as described in Refs.~\cite{berta2023gap,Berta2023,fang2022ultimate}.
As a result, the consequences of the generalized quantum Stein's lemma became no longer considered valid; most critically, the first---and so far the only known---general formulation of QRTs equipped with the second law has lost its validity, reopening the big question in quantum information theory~\cite{krueger2005open} regarding whether such a framework can be constructed at all~\cite{berta2023gap,Berta2023}.

In this work, we construct a proof of the generalized quantum Stein's lemma and thus redeem its consequences including those on the establishment of the second law of QRTs\@.
Our proof circumvents the gap of the existing analysis of the generalized quantum Stein's lemma in Ref.~\cite{Brandao2010} by developing alternative techniques for addressing the non-IIDness using a new continuity bound of the quantum relative entropy recently developed by Refs.~\cite{10206734,10129917}.
Below, we will introduce the framework of QRTs and present our main results on the generalized quantum Stein's lemma in this framework.
We will also discuss the consequences of these results to redeem the axiomatic formulation of QRTs equipped with the second law.

\paragraph*{Framework of QRTs.}
We present the framework of QRTs studied in this work based on the general framework in Ref.~\cite{Kuroiwa2020}.
We represent a quantum system by a finite-dimensional complex Hilbert space $\mathcal{H}=\mathbb{C}^d$ for some finite $d$.
The set of quantum states, i.e., positive semidefinite operators with unit trace, of the system $\mathcal{H}$ is denoted by $\mathcal{D}(\mathcal{H})$.
The identity operator is denoted by $\mathds{1}$.
A composite system is represented in terms of the tensor product; that is, for any $N\in\mathbb{N}$ with $\mathbb{N}\coloneqq\{1,2,\ldots\}$, a system composed of $N$ subsystems $\mathcal{H}$ is represented by $\mathcal{H}^{\otimes N}$.
For a state $\rho\in\mathcal{D}\qty(\mathcal{H}^{\otimes N})$ and any $n\in\{1,\ldots,N\}$, the partial trace over the $n$th subsystem is denoted by $\Tr_n\qty[\rho]$.
The set of quantum operations from an input system $\mathcal{H}_\mathrm{in}$ to an output system $\mathcal{H}_\mathrm{out}$ is represented by that of completely positive and trace-preserving (CPTP) linear maps, denoted by $\mathcal{C}\qty(\mathcal{H}_\mathrm{in}\to\mathcal{H}_\mathrm{out})$.

A QRT is specified by choosing a set of free operations as a subset of quantum operations for each pair of input system $\mathcal{H}_\mathrm{in}$ and output system $\mathcal{H}_\mathrm{out}$~\cite{Chitambar2018,Kuroiwa2020}, which we write as $\mathcal{O}\qty(\mathcal{H}_\mathrm{in}\to\mathcal{H}_\mathrm{out})\subseteq\mathcal{C}\qty(\mathcal{H}_\mathrm{in}\to\mathcal{H}_\mathrm{out})$.
A free state is defined as a state that can be generated from any initial state by a free operation~\cite{Kuroiwa2020}; more formally, for any system $\mathcal{H}$, we write the set of free states of $\mathcal{H}$ as $\mathcal{F}(\mathcal{H})\coloneqq\{\sigma\in\mathcal{D}\qty(\mathcal{H}):\forall\mathcal{H}^\prime,\forall\rho\in\mathcal{D}\qty(\mathcal{H}^\prime),\exists\mathcal{E}\in\mathcal{O}\qty(\mathcal{H}^\prime\to\mathcal{H})~\text{such that}~\sigma=\mathcal{E}\qty(\rho)\}$.
Following Refs.~\cite{Brand_o_2008,brandao2010reversible,Brandao2010,Brandao2015}, we work on QRTs with their sets of free states satisfying the following properties.
\begin{enumerate}
    \item \label{p1:main}For each finite-dimensional system $\mathcal{H}$, $\mathcal{F}(\mathcal{H})$ is closed and convex.
    \item \label{p2:main}For each finite-dimensional system $\mathcal{H}$, $\mathcal{F}(\mathcal{H})$ contains a full-rank free state $\sigma_{\mathrm{full}}>0$.
    \item \label{p3:main}For each finite-dimensional system $\mathcal{H}$ and any $N\in\mathbb{N}$, if $\rho\in\mathcal{F}\qty(\mathcal{H}^{\otimes N+1})$, then $\Tr_n\qty[\rho]\in\mathcal{F}\qty(\mathcal{H}^{\otimes N})$ for every $n\in\{1,\ldots,N+1\}$.
    \item \label{p4:main}For each finite-dimensional system $\mathcal{H}$ and any $N\in\mathbb{N}$, if $\rho\in\mathcal{F}\qty(\mathcal{H})$, then $\rho^{\otimes N}\in\mathcal{F}\qty(\mathcal{H}^{\otimes N})$.
    \item \label{p5:main}For each finite-dimensional system $\mathcal{H}$ and any $N\in\mathbb{N}$, if $\rho\in\mathcal{F}\qty(\mathcal{H}^{\otimes N})$, then $U_\pi \rho U_\pi^\dag\in\mathcal{F}\qty(\mathcal{H}^{\otimes N})$ for any permutation $\pi\in S_N$, where $S_N$ is the symmetric group of degree $N$, and $U_\pi$ for $\pi\in S_N$ is the unitary operator representing the permutation of the state of the $N$ subsystems of $\mathcal{H}^{\otimes N}$ according to $\pi$.
\end{enumerate}

Representative examples of QRTs, such as those of entanglement, coherence, and magic states, satisfy these properties~\cite{Chitambar2018}, and Refs.~\cite{Brand_o_2008,brandao2010reversible} worked on the theory of entanglement with these properties.
The original analysis of the generalized quantum Stein's lemma in Ref.~\cite{Brandao2010} assumes a stronger version of Property~\ref{p4:main} by requiring that for any free states $\rho$ and $\sigma$, the product state $\rho\otimes\sigma$ should be free; however, our version of Property~\ref{p4:main} suffices for the proof, leading to broader applicability (see Methods for details).
Property~\ref{p2:main} is not explicitly mentioned in Ref.~\cite{Brandao2015} but is necessary for proving and using the generalized quantum Stein's lemma.

A non-free state $\rho\in\mathcal{D}(\mathcal{H})\setminus\mathcal{F}(\mathcal{H})$ is called a resource state.
To quantify the amount of resource of a state, we consider a family of real functions $R_\mathcal{H}$ of the state of every system $\mathcal{H}$.
For brevity, we may omit the subscript $\mathcal{H}$ of $R_\mathcal{H}$ to write $R$ if obvious from the context.
The function $R$ is called a resource measure if $R$ quantifies the amount of resource without contradicting an intuition that free operations cannot increase resources, which is a property called monotonicity; i.e., for any free operation $\mathcal{E}$ and any state $\rho$, it should hold that $R(\rho)\geq R(\mathcal{E}(\rho))$.
Various types of resource measures have been proposed. 
For example, the relative entropy of resource is defined as~\cite{Chitambar2018,Kuroiwa2020}
\begin{equation}
    R_\mathrm{R}\qty(\rho)\coloneqq\min_{\sigma\in\mathcal{F}\qty(\mathcal{H})}\qty{D\left(\rho\middle|\middle|\sigma\right)},
\end{equation}
where $D\left(\rho\middle|\middle|\sigma\right)\coloneqq\Tr[\rho\log_2\rho]-\Tr[\rho\log_2\sigma]$ is the quantum relative entropy.
Its variant, the regularized relative entropy of resources, is defined as~\cite{Chitambar2018,Kuroiwa2020}
\begin{align}
\label{eq:regularized_relative_entropy_of_resource}
    R_\mathrm{R}^\infty\qty(\rho)&\coloneqq\lim_{N\to\infty}\frac{R_\mathrm{R}\qty(\rho^{\otimes N})}{N}\\
    &=\lim_{N\to\infty}\min_{\sigma\in\mathcal{F}\qty(\mathcal{H}^{\otimes N})}\qty{\frac{D\left(\rho^{\otimes N}\middle|\middle|\sigma\right)}{N}}.
\end{align}
Another example is the generalized robustness of resource (also known as global robustness) defined as~\cite{Chitambar2018}
\begin{equation}
\label{eq:generalized_robustness}
    R_\mathrm{G}\qty(\rho)\coloneqq\min_{\tau\in\mathcal{D}\qty(\mathcal{H})}\qty{s\geq 0:\frac{\rho+s\tau}{1+s}\eqqcolon\sigma\in\mathcal{F}(\mathcal{H})}.
\end{equation}
All these functions serve as resource measures, satisfying the monotonicity as required~\cite{Chitambar2018}.

A fundamental task in QRTs is the asymptotic conversion of quantum states.
For two quantum systems $\mathcal{H}_\mathrm{in}$ and $\mathcal{H}_\mathrm{out}$, and two states $\rho\in\mathcal{D}(\mathcal{H}_\mathrm{in})$ and $\sigma\in\mathcal{D}(\mathcal{H}_\mathrm{out})$, 
the asymptotic conversion from $\rho$ to $\sigma$ is a task of converting many copies of $\rho$ into as many copies of $\sigma$ as possible by a sequence of operations $\mathcal{E}_1,\mathcal{E}_2,\ldots$ within an asymptotically vanishing error.
Under a class $\mathcal{O}$ of operations, the conversion rate is defined as~\cite{Kuroiwa2020}
\begin{align}
\label{eq:conversion_rate}
    &r_\mathcal{O}\qty(\rho\to\sigma)\coloneqq\nonumber\\
    &\sup\left\{r\geq 0:\exists\qty{\mathcal{E}_N\in\mathcal{O}\qty(\mathcal{H}_\mathrm{in}^{\otimes N}\to\mathcal{H}_\mathrm{out}^{\otimes \lceil rN\rceil})}_{N\in\mathbb{N}},\right.\nonumber\\
    &\qquad\left.\liminf_{N\to\infty}\left\|\mathcal{E}_N\qty(\rho^{\otimes N})-\sigma^{\otimes \lceil rN\rceil}\right\|_1=0\right\},
\end{align}
where $\lceil{}\cdots{}\rceil$ is the ceiling function, $\|\cdots\|_1$ is the trace norm, and $\sigma^{\otimes 0}=1$.
A fundamental question in QRTs is whether one can establish a general formulation with an appropriate choice of the class $\mathcal{O}$ of operations and a single resource measure $R$ so that the resource measures $R(\rho)$ and $R(\sigma)$ should characterize the convertibility at rate $r_\mathcal{O}\qty(\rho\to\sigma)$, which is the second law of QRTs\@.

\paragraph*{Main result: Generalized quantum Stein's lemma.}---
Our main result is the proof of the generalized quantum Stein's lemma characterizing the optimal performance of a variant of quantum hypothesis testing to be performed in the above framework of QRTs\@.
In this variant of quantum hypothesis testing, as introduced in Ref.~\cite{Brandao2010}, we are initially given a parameter $N\in\mathbb{N}$, a classical description of a state $\rho\in\mathcal{D}\qty(\mathcal{H})$, and an unknown quantum state of the system $\mathcal{H}^{\otimes N}$.
The task is to perform a two-outcome measurement by a positive operator-valued measure (POVM) $\{E_N,\mathds{1}-E_N\}$ on $\mathcal{H}^{\otimes N}$ ($0\leq E_N\leq\mathds{1}$) to distinguish the following two cases.
\begin{itemize}
    \item Null hypothesis: The given state is $N$ IID copies $\rho^{\otimes N}$ of $\rho$.
    \item Alternative hypothesis: The given state is some free state $\sigma\in\mathcal{F}\qty(\mathcal{H}^{\otimes N})$ in the set satisfying Properties~\ref{p1:main}--\ref{p5:main}, where $\sigma$ may be in a non-IID form over $\mathcal{H}^{\otimes N}$.
\end{itemize}
If the measurement outcome is $E_N$, we will conclude that the given state was $\rho^{\otimes N}$, and if $\mathds{1}-E_N$, then was some free state $\sigma$ in the set $\mathcal{F}\qty(\mathcal{H}^{\otimes N})$ of free states.
For this hypothesis testing, we define the following two types of errors.
\begin{itemize}
    \item Type I error: The mistaken conclusion that the given state was some free state $\sigma\in\mathcal{F}\qty(\mathcal{H}^{\otimes N})$ when it was $\rho^{\otimes N}$, which happens with probability $\alpha_N\coloneqq\Tr\qty[\qty(\mathds{1}-E_N)\rho^{\otimes N})]$.
    \item Type II error: The mistaken conclusion that the given state was $\rho^{\otimes N}$ when it was some free state $\sigma_N\in\mathcal{F}\qty(\mathcal{H}^{\otimes N})$, which happens with probability $\beta_N\coloneqq\max_{\sigma\in\mathcal{F}(\mathcal{H}^{\otimes N})}\qty{\Tr\qty[E_N\sigma]}$ in the worst case.
\end{itemize}

By choosing appropriate POVMs, we can suppress the type II error exponentially in $N$  while keeping the type I error vanishingly small. 
The generalized quantum Stein's lemma characterizes the optimal exponent of the type II error when we require that the type I error should asymptotically vanish as $N$ increases.
Note that by choosing $\mathcal{F}(\mathcal{H}^{\otimes N})=\{\sigma^{\otimes N}\}$, the generalized quantum Stein's lemma reduces to quantum Stein's lemma for quantum hypothesis testing in the conventional setting~\cite{hiai1991proper,887855,Brandao2010}.
\begin{theorem}[\label{thm:main_generalized_steins_lemma}Generalized quantum Stein's lemma]
Given any family $\mathcal{F}$ of sets of free states satisfying Properties~\ref{p1:main}--\ref{p5:main} and any state $\rho\in\mathcal{D}(\mathcal{H})$, for any $\epsilon>0$,
there exists a sequence $\qty{\qty{E_N,\mathds{1}-E_N}}_{N\in\mathbb{N}}$ of POVMs
achieving
\begin{equation}
    \lim_{N\to\infty}\Tr\qty[\qty(\mathds{1}-E_N)\rho^{\otimes N}]=0
\end{equation}
and for each $N\in\mathbb{N}$ 
\begin{equation}
    -\frac{\log_2\qty(\max_{\sigma\in\mathcal{F}\qty(\mathcal{H}^{\otimes N})}\qty{\Tr\qty[E_N\sigma]})}{N}\geq R_\mathrm{R}^\infty(\rho)-\epsilon.
\end{equation}
Conversely, for any sequence $\qty{\qty{E_N,\mathds{1}-E_N}}_{N\in\mathbb{N}}$ of POVMs\@,
if there exists $\epsilon>0$ such that
\begin{equation}
    \liminf_{N\to\infty}-\frac{\log_2\qty(\max_{\sigma\in\mathcal{F}\qty(\mathcal{H}^{\otimes N})}\qty{\Tr\qty[E_N\sigma]})}{N}\geq R_\mathrm{R}^\infty(\rho)+\epsilon,
\end{equation}
then it holds that
\begin{equation}
    \lim_{N\to\infty}\Tr\qty[\qty(\mathds{1}-E_N)\rho^{\otimes N}]=1.
\end{equation}
\end{theorem}

The most formidable challenge in the proof of Theorem~\ref{thm:main_generalized_steins_lemma} appears in proving the former part of Theorem~\ref{thm:main_generalized_steins_lemma} on achievability (which is also called the direct part).
Indeed, most of the pages of Ref.~\cite{Brandao2010} were devoted to the analysis of the direct part while a relatively shorter proof of the latter converse part of Theorem~\ref{thm:main_generalized_steins_lemma} is validly given in Ref.~\cite{Brandao2010}.
In particular, the analysis of the direct part aims to bound the optimal achievable performance of distinguishing $\rho^{\otimes N}$ from a worst-case choice of free state $\sigma_N\in\mathcal{F}(\mathcal{H}^{\otimes N})$ in terms of the regularized relative entropy of resource $R_\mathrm{R}^\infty(\rho)$ in~\eqref{eq:regularized_relative_entropy_of_resource}.
The challenge arises from the non-IIDness of $\sigma_N$.
To address the non-IIDness, Ref.~\cite{Brandao2010} considered using an idea that symmetry of this task under permutation of the $N$ subsystems implies that almost all states of the $N$ subsystems are virtually identical and independent of each other~\cite{renner2006security,renner2007symmetry}, approximately recovering the IID structure of the state.
Then, Ref.~\cite{Brandao2010} tried to show that the approximation for recovering the IID structure does not change the quantum relative entropy up to a negligibly small amount, so that the overall bound of the optimal achievable performance should be given in terms of $R_\mathrm{R}^\infty(\rho)$ with the regularization, by invoking a technique in Lemma~III.9 of Ref.~\cite{Brandao2010} where the logical gap of the proof was found~\cite{berta2023gap}.
By contrast, our proof circumvents this gap of Ref.~\cite{Brandao2010} in dealing with the approximation for recovering the IID structure.
Instead, we develop alternative techniques using a new continuity bound of the quantum relative entropy recently discovered by Refs.~\cite{10206734,10129917} to properly deal with the approximation in addressing the non-IIDness\@.
See Methods for details.

\paragraph*{Redeeming the second law of QRTs.}
Our proof of the generalized quantum Stein's lemma redeems the second law of QRTs to characterize the asymptotic convertibility of general quantum resources, as originally attempted in Refs.~\cite{Brand_o_2008,brandao2010reversible,Brandao2010,Brandao2015}.
In~\eqref{eq:conversion_rate}, the class $\mathcal{O}$ of free operations was introduced in a conventional way, i.e., in a bottom-up approach by specifying what can be done.
By contrast, Refs.~\cite{Brand_o_2008,brandao2010reversible} formulated QRTs by introducing a complementary class of operations defined in an axiomatic approach by specifying what cannot be done.
In the setting of the asymptotic conversion, a fundamental requirement for the free operations is that the free operations should not generate resource states from free states.
There can be several different ways of introducing axioms to capture this requirement asymptotically, and even with such axioms, it is not always possible to formulate QRTs with the second law under the axiomatically defined class of operations ~\cite{Lami2023,PhysRevLett.119.180506}.
More precisely, the second law may hold for special types of quantum resources, e.g., for QRTs of coherence~\cite{PhysRevA.97.050301,berta2023gap} and athermality~\cite{PhysRevLett.111.250404}.
Similarly, the formulations with the second law may be possible using other variants of a composite quantum Stein's lemma~\cite{9031743,PhysRevLett.121.190503,hayashi2002optimal} while such formulations are not general enough to be applicable to entanglement~\cite{berta2023gap}.
More general theories with the second law may be formulated by going beyond the law of quantum mechanics, e.g., by allowing post-selection~\cite{regula2023reversibility} or non-physical quasi-operations~\cite{wang2023reversible}.
However, in full generality and within the law of quantum mechanics, the only known way for the desired formulation is that of Refs.~\cite{Brand_o_2008,brandao2010reversible,Brandao2010,Brandao2015} using the generalized quantum Stein's lemma.
The axiomatic class of operations
for this formulation is called asymptotically resource non-generating operations, given in terms of the generalized robustness of resource $R_\mathrm{G}$ in~\eqref{eq:generalized_robustness}, for any $\mathcal{H}_\mathrm{in}$, $\mathcal{H}_\mathrm{out}$, $N\in\mathbb{N}$, and some choice of a sequence $\{\epsilon_N>0\}_{N\in\mathbb{N}}$ of error parameters vanishing as $\lim_{N\to\infty}\epsilon_N=0$, by
\begin{align}
\label{eq:asymptotically_resource_non_generating_operations}
    &\tilde{\mathcal{O}}\qty(\mathcal{H}_\mathrm{in}^{\otimes N}\to\mathcal{H}_\mathrm{out}^{\otimes N})\coloneqq\left\{\mathcal{E}_N\in\mathcal{C}\qty(\mathcal{H}_\mathrm{in}^{\otimes N}\to\mathcal{H}_\mathrm{out}^{\otimes N}):\right.\nonumber\\
    &\quad\left.\forall\rho\in\mathcal{F}\qty(\mathcal{H}_\mathrm{in}),R_\mathrm{G}\qty(\mathcal{E}_N\qty(\rho^{\otimes N}))\leq\epsilon_N\right\}.
\end{align}

Under the asymptotically resource non-generating operations $\tilde{\mathcal{O}}$ in~\eqref{eq:asymptotically_resource_non_generating_operations}, using the generalized quantum Stein's lemma in Theorem~\ref{thm:main_generalized_steins_lemma}, we can characterize the asymptotic convertibility of general quantum resources, as shown in Ref.~\cite{Brandao2015}, by 
\begin{equation}
\label{eq:second_law}
    r_{\tilde{O}}\qty(\rho\to\sigma)=\frac{R_\mathrm{R}^\infty(\rho)}{R_\mathrm{R}^\infty(\sigma)},
\end{equation}
granting the role of entropy in thermodynamics to the regularized relative entropy of resource $R_\mathrm{R}^\infty$ in the general QRTs\@.
Thus, we have the second law of QRTs for the general types of quantum resources satisfying Properties~\ref{p1:main}--\ref{p5:main}, formally analogous to the axiomatic formulation of the second law of thermodynamics.
As in the idealization of quasi-static adiabatic operations axiomatically introduced for the second law of thermodynamics~\cite{LIEB19991,Lieb2004,10.1063/1.883034}, it may be in general unknown how to realize all operations in the axiomatically defined class.
However, what is physically realizable should always be in this axiomatic class of operations regardless of technological advances in the future, in the same way as thermodynamics.

Lastly, in the entanglement theory, $\mathcal{F}(\mathcal{H})$ is taken as the set of separable states on two spatially separated systems $A$ and $B$, i.e., those represented as a convex combination of pure product states $\sum_{j}p(j)\ket{\psi_j}\bra{\psi_j}^A\otimes\ket{\phi_j}\bra{\phi_j}^B$ of the system $\mathcal{H}=\mathcal{H}^A\otimes\mathcal{H}^B$~\cite{Horodecki2009}.
In this case, a fundamental resource state is an ebit, i.e., a two-qubit maximally entangled state $\Phi\coloneqq\ket{\Phi}\bra{\Phi}$ with $\ket{\Phi}\coloneqq\frac{1}{\sqrt{2}}\qty(\ket{0}^A\otimes\ket{0}^B+\ket{1}^A\otimes\ket{1}^B)$.
For an ebit, the regularized relative entropy of entanglement is $R_\mathrm{R}^\infty(\Phi)=1$~\cite{Horodecki2009}.
Given any state $\rho\in\mathcal{D}(\mathcal{H})$, in the asymptotic conversion from $\rho$ to $\Phi$, the maximum number of ebits $\Phi$ obtained per $\rho$ is called the distillable entanglement~\cite{PhysRevA.53.2046}, written under the class $\tilde{O}$ of operations as $E_{\mathrm{D},\tilde{\mathcal{O}}}(\rho)\coloneqq r_{\tilde{\mathcal{O}}}\qty(\rho\to\Phi)$.
Also in the asymptotic conversion rate from $\Phi$ to $\rho$, the minimum required number of ebits $\Phi$ per $\rho$ is called the entanglement cost~\cite{PhysRevA.53.2046,Hayden_2001,Yamasaki2024}, given under $\tilde{\mathcal{O}}$ by
$E_{\mathrm{C},\tilde{\mathcal{O}}}(\rho)\coloneqq 1/{r_{\tilde{\mathcal{O}}}\qty(\Phi\to\rho)}$~\cite{Kuroiwa2020}.
Due to~\eqref{eq:second_law}, the generalized quantum Stein's lemma shows the asymptotically reversible interconvertibility between all, pure and mixed, bipartite entangled states, as originally intended in Refs.~\cite{Brand_o_2008,brandao2010reversible}, i.e.,
\begin{equation}
    E_{\mathrm{D},\tilde{\mathcal{O}}}(\rho)=E_{\mathrm{C},\tilde{\mathcal{O}}}(\rho)=R_\mathrm{R}^\infty(\rho).
\end{equation}
resolving the question raised in Ref.~\cite{krueger2005open}.

\paragraph*{Outlook.}
We have constructed a proof of the generalized quantum Stein's lemma presented in Theorem~\ref{thm:main_generalized_steins_lemma}, by circumventing the logical gap of the original analysis in Ref.~\cite{Brandao2010} pointed out by Ref.~\cite{berta2023gap}.
Our proof avoids the gap in Ref.~\cite{Brandao2010} by developing alternative techniques for addressing the non-IIDness of states appearing in the generalized quantum Stein's lemma, using a new continuity bound of the quantum relative entropy recently discovered by Refs.~\cite{10206734,10129917}.
Our proof of the generalized quantum Stein's lemma redeems the axiomatic formulation of QRTs equipped with the second law under Properties~\ref{p1:main}--\ref{p5:main}.
As originally intended in Refs.~\cite{Brand_o_2008,brandao2010reversible,Brandao2010,Brandao2015}, the formulation of QRTs with the second law revived here is an indication of the success of thermodynamics, now in the realm of quantum information theory.
Since Refs.~\cite{Brand_o_2008,brandao2010reversible,Brandao2010,Brandao2015} were out, QRTs have found much broader applicability beyond entanglement~\cite{Chitambar2018,Kuroiwa2020}; in view of this, a remaining question would be how universally these results may be extended beyond the scope of convex and finite-dimensional QRTs satisfying Properties~\ref{p1:main}--\ref{p5:main}, e.g., to non-convex QRTs~\cite{Kuroiwa2020,PhysRevA.104.L020401,PRL,PRA} and infinite-dimensional QRTs~\cite{Kuroiwa2020,PhysRevA.104.L020401,Regula2021,Lami2021,ferrari2023asymptotic,Yamasaki2024}.
Toward such an overarching goal, our work provides a solid formulation and powerful tools to study broad future possibilities for the use of quantum resources.

\begin{acknowledgments}
  H.Y.\ acknowledges Kaito Watanabe for discussion. 
  K.K.\ was supported by a Funai Overseas Scholarship and a Perimeter Residency Doctoral Award. 
  This work was supported by JST PRESTO Grant Number JPMJPR201A, JPMJPR23FC, JSPS KAKENHI Grant Number JP23K19970, and MEXT Quantum Leap Flagship Program (MEXT QLEAP) JPMXS0118069605, JPMXS0120351339\@.
\end{acknowledgments}

\section*{Author contributions}

Both authors contributed to the conception of the work, the analysis and interpretation in the work, and the preparation of the manuscript.

\section*{Competing interests}

The authors declare no competing interests.

\section*{Additional information}

Supplementary Information is available for this paper.
Correspondence and requests for materials should be addressed to Hayata Yamasaki.

\section*{Methods}
In Methods, we illustrate our proof of the generalized quantum Stein's lemma (Theorem~\ref{thm:main_generalized_steins_lemma} in the main text). 
Before going into details, we point out that, viewed with hindsight based on our proof, the valid parts of the proof strategy described in the original analysis of the generalized quantum Stein's lemma in Ref.~\cite{Brandao2010} are indeed useful for constructing our proof if we modify the statements appropriately; however, this is nontrivial, in view of many other possible proof strategies suggested in Ref.~\cite{berta2023gap}.
After all, when a logical gap is found in a mathematical proof of a theorem, the overall logic to prove the theorem loses its validity.
In a fortunate case, similar lines of logic may eventually lead to some corrected proof;
however, in the worst case, complete proof may never be obtained in the same strategy due to the hindrance arising from the insoluble gap.
In the latter case, one might need to construct a completely different logic from scratch or even doubt whether the theorem to be proven should hold true in the first place.
Thus, when Ref.~\cite{berta2023gap} pointed out the logical gap of the analysis in Ref.~\cite{Brandao2010}, it was natural to discuss multiple different strategies for possible proofs including those starting over almost from scratch.
In the description of our proof, we will clarify which part of the analysis in Ref.~\cite{Brandao2010} can be reused in our proof and which cannot.

Our strategy for the overall proof is summarized as follows. 
We will first reduce the proof of Theorem~\ref{thm:main_generalized_steins_lemma} to that of a simpler statement (Theorem~\ref{thm:main_actual_statement_to_prove}) in a slightly different way from that in Ref.~\cite{Brandao2010} so as to circumvent the logical gap of Ref.~\cite{Brandao2010}.
Our goal is to prove this simpler statement. 
This statement is further divided into two parts: the direct part and the strong converse, similar to Ref.~\cite{Brandao2010}. 
The strong converse (Theorem~\ref{thm:main_strong_converse}) is the same as a validly proven statement in Ref.~\cite{Brandao2010}, which can be reused in our proof. 
Our main contribution is to construct a proof of the direct part (Theorem~\ref{thm:main_direct_part}) while Ref.~\cite{Brandao2010} has a logical gap in proving the direct part as pointed out in Ref.~\cite{berta2023gap}.

Regarding the notations, for a Hermitian operator $A$ with spectral decomposition $A=\sum_j\lambda_j\ket{j}\bra{j}$, we write
\begin{equation}
    \qty(A)_+\coloneqq\sum_{j:\lambda_j>0}\ket{j}\bra{j}.
\end{equation}
For two functions $f,g$ and $N\in\mathbb{N}$, $f(N)=O(g(N))$ means $\limsup_{N\to\infty}|f(N)|/g(N)<\infty$, and $f(N)=o(g(N))$ means $\lim_{n\to\infty}f(N)/g(N)=0$.
See also Refs.~\cite{N4,Wilde_2017,watrous_2018} for other conventional notations in quantum information theory.

In the following, we will first describe the reduction of the proof of Theorem~\ref{thm:main_generalized_steins_lemma} to that of a simpler statement.
Then, we will divide the statement into the strong converse and direct part to describe the proof of each of these parts.

\subsection{Reduction of generalized quantum Stein's lemma to simpler statement}
\label{sec:reduction}

To prove Theorem~\ref{thm:main_generalized_steins_lemma}, we will show Theorem~\ref{thm:main_actual_statement_to_prove} below, which is a modified version of Proposition III.1 of Ref.~\cite{Brandao2010}. 
Note that in Proposition III.1 of Ref.~\cite{Brandao2010}, the condition $y > 0$ was not imposed while Theorem~\ref{thm:main_actual_statement_to_prove} for our proof has a slightly stronger assumption due to imposing $y > 0$. 
We will reduce our proof of Theorem~\ref{thm:main_generalized_steins_lemma} to proving Theorem~\ref{thm:main_actual_statement_to_prove} with this condition, which is necessary for our proof to circumvent the logical gap of Ref.~\cite{Brandao2010}. 

\begin{theorem}[Modified version of Proposition III.1 of Ref.~\cite{Brandao2010}]
    \label{thm:main_actual_statement_to_prove}
    For any family $\{\mathcal{F}\qty(\mathcal{H}^{\otimes N})\}_{N\in\mathbb{N}}$ of sets of free states satisfying Properties~\ref{p1}-\ref{p5}, any state $\rho \in \mathcal{D}(\mathcal{H})$, and any $y > 0$, it holds that
    \begin{equation}
            \label{eq:main_limit_steins_lemma_actual}
        \lim_{N\to\infty} \min_{\sigma \in \mathcal{F}(\mathcal{H}^{\otimes N})} \Tr\qty[\qty(\rho^{\otimes N} - 2^{yN} \sigma)_{+}] 
        = 
        \begin{cases}
            0, & y > R_R^{\infty}(\rho), \\ 
            1, & y < R_R^{\infty}(\rho).
        \end{cases}
    \end{equation}
\end{theorem}

Here, we show that Theorem~\ref{thm:main_generalized_steins_lemma} indeed follows from Theorem~\ref{thm:main_actual_statement_to_prove}.
We basically follow the proof of Theorem~I in Ref.~\cite{Brandao2010}, but in our case, we need to keep it in mind that the condition $y > 0$ is imposed in Theorem~\ref{thm:main_actual_statement_to_prove}. 

To prove Theorem~\ref{thm:main_generalized_steins_lemma} from Theorem~\ref{thm:main_actual_statement_to_prove},
we introduce a notation $p_N(\eta,K)$, with $N \in \mathbb{N}$, a state $\eta$, and a positive number $K > 0$, to represent the solution of a convex optimization problem as
\begin{align}
\label{eq:alpha_N}
        &p_N(\eta,K) \nonumber \\ 
        &\coloneqq \max \left\{\Tr\qty[E\eta] : 0 \leq E \leq \mathds{1}, \,\max_{\sigma \in \mathcal{F}(\mathcal{H}^{\otimes N})}\Tr\qty[E\sigma] \leq \frac{1}{K}\right\}.
\end{align}
With this notation, $p_N(\rho^{\otimes N},2^{yN})$ represents the maximum success probability $1-\alpha_N$ of avoiding the type I error under the condition that the type II error should decay with $N$ at the rate $y$, i.e., $\beta_N\leq 2^{-yN}$.
By definition, the optimization problem in~\eqref{eq:alpha_N} always has some feasible solution, and we have $0\leq p_N(\eta,K) \leq 1$ for any state $\eta$ and any $K > 0$.
We employ the dual problem of the optimization problem in the definition of $p_N$ in~\eqref{eq:alpha_N}; in particular, as shown in the proof of Theorem~I in Ref.~\cite{Brandao2010}, we have 
\begin{align}
        &p_N(\rho^{\otimes N},2^{yN})
        =\nonumber\\
        &\min\Bigg\{\Tr\qty[\qty(\rho^{\otimes N} - b\sigma)_{+}] + \frac{b}{2^{yN}}:
        \sigma \in \mathcal{F}(\mathcal{H}^{\otimes N}),\,b\geq 0 \Bigg\}. 
\end{align}

The direct part of Theorem~\ref{thm:main_generalized_steins_lemma} follows from Theorem~\ref{thm:main_actual_statement_to_prove} as follows. 
For any $\epsilon > 0$ and any state $\rho$, if  $R_\mathrm{R}^\infty(\rho) - \epsilon \leq 0$,
the sequence $\{\{E_N,\mathds{1}-E_N\}\}_{N\in\mathbb{N}}$ of POVMs with $E_N=\mathds{1}$ for all $N$ satisfies the requirements of Theorem~\ref{thm:main_generalized_steins_lemma}. 
In the following, we analyze the other case, i.e., 
\begin{equation}
\label{eq:y_smaller_than_R}
    y \coloneqq R_{R}^{\infty}(\rho) - \epsilon>0;
\end{equation}
in this case, for each $N\in\mathbb{N}$, let $E_N$ be the optimal operator $E$
in the optimization problem in~\eqref{eq:alpha_N} for $p_N(\rho^{\otimes N},2^{yN})$, which achieves a type II error bounded by 
\begin{equation}
\label{eq:type_II_bound}
    -\frac{\log_2\qty(\max_{\sigma\in\mathcal{F}\qty(\mathcal{H}^{\otimes N})}\qty{\Tr\qty[E_N\sigma]})}{N}\geq R_\mathrm{R}^\infty(\rho)-\epsilon,
\end{equation}
due to the constraint in~\eqref{eq:alpha_N}, i.e.,
\begin{equation}
    \max_{\sigma \in \mathcal{F}(\mathcal{H}^{\otimes N})}\Tr\qty[E_N\sigma] \leq \frac{1}{2^{yN}}.
\end{equation}
Since $p_N(\rho^{\otimes N},2^{yN}) \leq 1$ and $\Tr\qty[\qty(\rho^{\otimes N} - b\sigma)_{+}]\geq 0$ for any $\sigma$, the optimal $b$ must satisfy $b \leq 2^{yN}$. 
Therefore, we have 
\begin{align}
        &p_N(\rho^{\otimes N},2^{yN}) 
        \geq \nonumber \\ 
        &\min\Bigg\{\Tr\qty[\qty(\rho^{\otimes N} - 2^{N\qty(R_{R}^{\infty}(\rho) - \epsilon)}\sigma)_{+}]:\sigma \in \mathcal{F}(\mathcal{H}^{\otimes N}) \Bigg\}.
\end{align}
Due to $y>0$ in~\eqref{eq:y_smaller_than_R},
using Theorem~\ref{thm:main_actual_statement_to_prove}, we obtain
\begin{equation}
    \lim_{N\to\infty} \Tr\qty[\qty(\rho^{\otimes N} - 2^{N\qty(R_{R}^{\infty}(\rho) - \epsilon)}\sigma)_{+}] = 1.
\end{equation}
Therefore, with this optimal choice of POVMs $\{\{E_N,\mathds{1}-E_N\}\}_{N\in\mathbb{N}}$, we have 
\begin{equation}
    \lim_{N\to\infty} p_N(\rho^{\otimes N},2^{yN})= 1,
\end{equation}
indicating that the type I error converges to $0$, i.e.,
\begin{equation}
    \lim_{N\to\infty}\Tr\qty[\qty(\mathds{1}-E_N)\rho^{\otimes N}]=0,
\end{equation}
which with the bound of the type II error in~\eqref{eq:type_II_bound} achieves the requirements of Theorem~\ref{thm:main_generalized_steins_lemma}.

Conversely,
for some $\epsilon>0$,
consider any sequence $\{\{E_N,\mathds{1}-E_N\}\}_{N\in\mathbb{N}}$ of POVMs
satisfying
\begin{equation}
\label{eq:constraint_converse}
    \liminf_{N\to\infty}-\frac{\log_2\qty(\max_{\sigma\in\mathcal{F}\qty(\mathcal{H}^{\otimes N})}\qty{\Tr\qty[E_N\sigma]})}{N}\geq R_\mathrm{R}^\infty(\rho)+\epsilon,
\end{equation}
as required in Theorem~\ref{thm:main_generalized_steins_lemma};
in this case, for any $\tilde{\epsilon}$ satisfying $0<\tilde{\epsilon}<\epsilon$, take
\begin{equation}
\label{eq:y_larger_than_R}
    y \coloneqq R_{R}^{\infty}(\rho) + \tilde{\epsilon}>0.
\end{equation}
Due to~\eqref{eq:constraint_converse},
there exists $N_0\in\mathbb{N}$ such that for every $N\geq N_0$,  the operator $E_N$ in this sequence $\{\{E_N,\mathds{1}-E_N\}\}_{N\in\mathbb{N}}$ of POVMs should satisfy the constraint in the optimization problem of~\eqref{eq:alpha_N} for $p_N(\rho^{\otimes N},2^{yN})$;
then, let $E$ be any of such operators $E_N$, satisfying
\begin{equation}
    \max_{\sigma \in \mathcal{F}(\mathcal{H}^{\otimes N})}\Tr\qty[E\sigma] \leq \frac{1}{2^{yN}}.
\end{equation}
By choosing $b = 2^{N\qty(R_{R}^{\infty}(\rho) + \frac{\tilde{\epsilon}}{2})}$,
we have 
\begin{align}
        &p_N(\rho^{\otimes N},2^{yN}) \nonumber \\ 
        &\leq \min\left\{\Tr\qty[\qty(\rho^{\otimes N} - 2^{N\qty(R_{R}^{\infty}(\rho) + \frac{\tilde{\epsilon}}{2})}\sigma)_{+}]+ 2^{-N\frac{\tilde{\epsilon}}{2}}:\right. \nonumber \\ 
        &\quad \left.\sigma \in \mathcal{F}(\mathcal{H}^{\otimes N}) \right\}. 
\end{align}
Due to $y>0$ in~\eqref{eq:y_larger_than_R}, using Theorem~\ref{thm:main_actual_statement_to_prove}, we obtain
\begin{equation}
    \lim_{N \to \infty}\min_{\sigma \in \mathcal{F}(\mathcal{H}^{\otimes N})} \Tr\qty[\qty(\rho^{\otimes N} - 2^{N\qty(R_{R}^{\infty}(\rho) + \frac{\tilde{\epsilon}}{2})}\sigma)_{+}] = 0, 
\end{equation}
and we also have $2^{-N\frac{\tilde{\epsilon}}{2}} \to 0$ as $N\to\infty$. 
Therefore, 
\begin{equation}
    \lim_{N\to\infty} p_N(\rho^{\otimes N},2^{yN})= 0, 
\end{equation}
indicating that the type I error converges to $1$, i.e.,
\begin{equation}
    \lim_{N\to\infty}\Tr\qty[\qty(\mathds{1}-E_N)\rho^{\otimes N}]=1.
\end{equation}

\subsection{Proof of strong converse and direct part}
\label{sec:converse_direct}

In this section, we outline the proof of Theorem~\ref{thm:main_actual_statement_to_prove}.
As in Ref.~\cite{Brandao2010}, we divide the statement of Theorem~\ref{thm:main_actual_statement_to_prove} into two parts: the strong converse and the direct part.

\subsubsection{Strong converse}
The strong converse guarantees that the limit in~\eqref{eq:main_limit_steins_lemma_actual} converges to zero when $y > R_R^{\infty}(\rho)$ while the limit has a nonzero value when $y < R_R^{\infty}(\rho)$. 
This was validly shown in Ref.~\cite{Brandao2010}, and we can reuse the statement as shown below.
In Ref.~\cite{Brandao2010}, Theorem~\ref{thm:main_strong_converse} was shown with the stronger version of Property~\ref{p4}, where $\rho \otimes \sigma$ has to be free for all free states $\rho$ and $\sigma$. 
However, the analysis in Ref.~\cite{Brandao2010} only used our version of Property~\ref{p4}, where $\rho^{\otimes N}$ is free for any free state $\rho$ and every positive integer $N$. 
Thus, the proof of the strong converse follows even with our version of Property~\ref{p4}. 
See Corollary~III.3 of Ref.~\cite{Brandao2010} for the proof.

\begin{theorem}[Strong converse of Theorem~\ref{thm:main_actual_statement_to_prove}, Collorary~III.3 of Ref.~\cite{Brandao2010}]
\label{thm:main_strong_converse}
Given any family $\{\mathcal{F}\qty(\mathcal{H}^{\otimes N})\}_{N\in\mathbb{N}}$ of sets of free states satisfying Properties~\ref{p1}--\ref{p5} and any state $\rho\in\mathcal{D}\qty(\mathcal{H})$, for any
\begin{equation}
    y>\lim_{N\to\infty}\min_{\sigma\in\mathcal{F}\qty(\mathcal{H}^{\otimes N})}\qty{\frac{D\left(\rho^{\otimes N}\middle|\middle|\sigma\right)}{N}},
\end{equation}
it holds that
\begin{equation}
    \lim_{N\to\infty}\min_{\sigma\in\mathcal{F}\qty(\mathcal{H}^{\otimes N})}\qty{\Tr\qty[\qty(\rho^{\otimes N}-2^{yN}\sigma)_+]}=0.
\end{equation}
Also for any
\begin{equation}
    y<\lim_{N\to\infty}\min_{\sigma\in\mathcal{F}\qty(\mathcal{H}^{\otimes N})}\qty{\frac{D\left(\rho^{\otimes N}\middle|\middle|\sigma\right)}{N}},
\end{equation}
it holds that
\begin{equation}
    \liminf_{N\to\infty}\min_{\sigma\in\mathcal{F}\qty(\mathcal{H}^{\otimes N})}\qty{\Tr\qty[\qty(\rho^{\otimes N}-2^{yN}\sigma)_+]}>0.
\end{equation}
\end{theorem}

\subsubsection{Main contribution: Direct part}
The direct part shows that the limit in~\eqref{eq:main_limit_steins_lemma_actual} becomes $1$ for $y < R_{R}^{\infty}(\rho)$, while the strong converse in Theorem~\ref{thm:main_strong_converse} only claims that the limit should be nonzero. 
Our main contribution in this work is to construct a proof of the direct part while the original analysis of the direct part of the general quantum Stein's lemma in Ref.~\cite{Brandao2010} has a logical gap as pointed out in Ref.~\cite{berta2023gap}. 

\begin{theorem}[Direct part of Theorem~\ref{thm:main_actual_statement_to_prove}]
\label{thm:main_direct_part}
For any family $\{\mathcal{F}\qty(\mathcal{H}^{\otimes N})\}_{N\in\mathbb{N}}$ of sets of free states satisfying Properties~\ref{p1}-\ref{p5}, any state $\rho\in\mathcal{D}\qty(\mathcal{H})$, and any $y>0$, if it holds that
\begin{equation}
    y< \lim_{N\to\infty}\min_{\sigma\in\mathcal{F}\qty(\mathcal{H}^{\otimes N})}\frac{D\left(\rho^{\otimes N}\middle|\middle|\sigma\right)}{N},
\end{equation} 
then we have 
\begin{equation}
\label{eq:main_limit_unity}
\liminf_{N\to\infty}\min_{\sigma\in\mathcal{F}\qty(\mathcal{H}^{\otimes N})}\qty{\Tr\qty[\qty(\rho^{\otimes N}-2^{yN}\sigma)_+]} = 1.
\end{equation}
\end{theorem}
Here, we highlight our main ideas to prove Theorem~\ref{thm:main_direct_part}. 
The details of our proof are given in Supplementary Information.
In the proof, we will show the contrapositive of the statement; that is, we will show for $y>0$, if 
\begin{equation}
\label{eq:main_non_zero_limit}
\liminf_{N\to\infty}\min_{\sigma\in\mathcal{F}\qty(\mathcal{H}^{\otimes N})}\qty{\Tr\qty[\qty(\rho^{\otimes N}-2^{yN}\sigma)_+]} \in (0,1),
\end{equation}
then
\begin{equation}
    \label{eq:main_y_geq_regularized_relative_entropy}
    y\geq \lim_{N\to\infty}\min_{\sigma\in\mathcal{F}\qty(\mathcal{H}^{\otimes N})}\qty{\frac{D\left(\rho^{\otimes N}\middle|\middle|\sigma\right)}{N}}.
\end{equation} 

In the first step of our proof, in the same way as Ref.~\cite{Brandao2010}, for the sequence $\{\sigma_N \in \mathcal{F}(\mathcal{H}^{\otimes N})\}_{N \in \mathbb{N}}$ of optimal free states in the minimization of~\eqref{eq:main_non_zero_limit},
we find a state $\rho_N \in \mathcal{D}(\mathcal{H}^{\otimes N})$ that has a non-negligible fidelity to $\rho^{\otimes N}$ for each $N$ and satisfies 
\begin{equation}
    \label{eq:main_inequality_1}
    \rho_N \leq 2^{yN + o(N)} \sigma_N. 
\end{equation}
We can take $\sigma_N$ and $\rho_N$ as permutation-invariant states, but $\sigma_N$ and $\rho_N$ may be in non-IID forms, from which the challenge arises. 

In the second step, 
to address the non-IIDness, similar to Ref.~\cite{Brandao2010}, we use the idea that the symmetry under the permutation of the $N$ subsystems implies that almost all states of the $N$ subsystems are virtually identical and independent of each other~\cite{renner2006security,renner2007symmetry}, which approximately recovers the IID structure of the state. 
For this purpose, Lemma~III.5 of Ref.~\cite{Brandao2010} represented this approximation by combining an operator inequality and the trace distance, but for our proof, it is crucial to represent this approximation solely in terms of a single operator inequality.
In particular, by tracing out a sublinear number $o(N)$ of subsystems, we relate the permutation-invariant state $\rho_N$ to the IID state $\rho^{\otimes N}$ in terms of such a single operator inequality; 
more formally, for some state $\tilde{\Delta}_{N-o(N)} \in \mathcal{D}(\mathcal{H}^{\otimes N-o(N)})$ and sequence $\{\epsilon_N>0\}_{N\in\mathbb{N}}$ of asymptotically vanishing parameters $\lim_{N\to\infty}\epsilon_N=0$, our operator inequality reads
\begin{align}
    \label{eq:main_inequality_4}
        \rho^{\otimes N-o(N)}\leq 2^{o(N)} \qty(\Tr_{1,2,\ldots,o(N)}\qty[\rho_N] + c_N\tilde{\Delta}_{N-o(N)}), 
\end{align}
where $c_N$ is a coefficient of order $c_N=O\qty(\epsilon_N 2^{yN})$, and $\Tr_{1,2,\ldots}$ is the partial trace over subsystems $1,2,\ldots$ out of the $N$ subsystems (See also Supplementary Information for the definitions of $\tilde{\Delta}_{N-o(N)}$, $\epsilon_N$, and $c_N$).
Using~\eqref{eq:main_inequality_1} and~\eqref{eq:main_inequality_4}, we construct a single operator inequality 
\begin{equation}
        \label{eq:main_inequality_2}
    \rho^{\otimes N - o(N)} \leq 2^{yN + o(N)}\tilde{\sigma}_{N-o(N)}, 
\end{equation}
with a state $\tilde{\sigma}_{N-o(N)}$ in the form of
\begin{equation}
    \label{eq:asym_free_state}
    \tilde{\sigma}_{N-o(N)} \coloneqq \frac{\Tr_{1,2,\ldots,o(N)}\qty[\sigma_{N}]+\frac{\epsilon_N}{2}\tilde{\Delta}_{N-o(N)}}{1 + \frac{\epsilon_N}{2}}. 
\end{equation}
Then, we show that the operator inequality~\eqref{eq:main_inequality_2} can be converted into a bound of the quantum relative entropy 
\begin{equation}
    \label{eq:relative_entropy_bound_tilde}
    D\left(\rho^{\otimes N - o(N)}\middle|\middle| \tilde{\sigma}_{N-o(N)}\right) \leq yN + o(N).  
\end{equation}

In the third step, 
we develop a technique for properly dealing with the difference between the quantum relative entropies with respect to a free state $\sigma$ in~\eqref{eq:main_y_geq_regularized_relative_entropy} and the state $\tilde{\sigma}_{N-o(N)}$ in~\eqref{eq:relative_entropy_bound_tilde} obtained from our approximation, which does not appear in the analysis of Ref.~\cite{Brandao2010}.
To this goal, we introduce a concept of asymptotically free states.
We define a set of asymptotically free states as
\begin{align}
        \label{eq:main_asymptotically_free_states}
&\mathcal{F}^{(\epsilon_N)}\qty(\mathcal{H}^{\otimes N})\nonumber\\ 
&\coloneqq\qty{\tilde{\sigma}\in\mathcal{D}\qty(\mathcal{H}^{\otimes N}):\exists\sigma\in\mathcal{F}\qty(\mathcal{H}^{\otimes N}),\,\|\tilde{\sigma}-\sigma\|_1\leq\epsilon_N},
\end{align}
which is the set of free states up to approximation within $\epsilon_N$ in terms of the trace distance.
With this definition, we show that $\tilde{\sigma}_{N - o(N)}$ in~\eqref{eq:asym_free_state} is an asymptotically free state, i,e.,
\begin{equation}
    \tilde{\sigma}_{N - o(N)}\in\mathcal{F}^{(\epsilon_N)}\qty(\mathcal{H}^{\otimes N}).
\end{equation}  
By taking the minimum over $\tilde{\sigma} \in \mathcal{F}^{(\epsilon_N)}(\mathcal{H}^{\otimes N - o(N)})$, it follows from~\eqref{eq:relative_entropy_bound_tilde} that
\begin{equation}
    \label{eq:approximately_free_bound}
\min_{\tilde{\sigma}\in\mathcal{F}^{(\epsilon_N)}\qty(\mathcal{H}^{\otimes N - o(N)})}\qty{\frac{D\left(\rho^{\otimes N-o(N)}\middle|\middle|\tilde{\sigma}\right)}{N - o(N)}} \leq y + o(1), 
\end{equation}
leading to 
\begin{equation}
    \label{eq:main_inequality_3}
    y\geq \limsup_{N\to\infty} \min_{\tilde{\sigma}\in\mathcal{F}^{(\epsilon_N)}\qty(\mathcal{H}^{\otimes N})}\qty{\frac{D\left(\rho^{\otimes N}\middle|\middle|\tilde{\sigma}\right)}{N}}. 
\end{equation}
This inequality almost shows what we aim to prove, i.e.,~\eqref{eq:main_y_geq_regularized_relative_entropy}; however, a difference remains in that the minimization of the quantum relative entropy in~\eqref{eq:main_inequality_3} is taken with respect to the set of asymptotically free states $\mathcal{F}^{(\epsilon_N)}\qty(\mathcal{H}^{\otimes N})$, rather than the set of free states $\mathcal{F}\qty(\mathcal{H}^{\otimes N})$ required for the regularized relative entropy of resource in~\eqref{eq:main_y_geq_regularized_relative_entropy}.

In the final step, we build up a technique to address this difference by using a continuity bound of the quantum relative entropy recently discovered in Refs.~\cite{10206734,10129917}.
Using this continuity bound, we identify the upper and lower bounds of the quantity appearing on the left-hand side of~\eqref{eq:main_inequality_3}
\begin{align}
    \label{eq:main_looser_bound_lower}
        &\min_{\sigma \in \mathcal{F}\qty(\mathcal{H}^{\otimes N})}\frac{D\left(\rho^{\otimes N}\middle|\middle|\sigma\right)}{N} - O\qty(\frac{\qty(\log_2\qty(\frac{1}{\epsilon_N})+N)^2\sqrt{\epsilon_N}}{N})\nonumber \\
        &\leq \min_{\tilde{\sigma} \in \mathcal{F}^{(\epsilon_N)}\qty(\mathcal{H}^{\otimes N})}\frac{D\left(\rho^{\otimes N}\middle|\middle|\tilde{\sigma}\right)}{N} \\ 
    \label{eq:main_looser_bound_upper}
        &\leq \min_{\sigma \in \mathcal{F}\qty(\mathcal{H}^{\otimes N})}\frac{D\left(\rho^{\otimes N}\middle|\middle|\sigma\right)}{N}.
\end{align}
We further show that under our assumption of
\begin{equation}
    y>0,
\end{equation}
the error parameter $\epsilon_N$ introduced in~\eqref{eq:main_inequality_4} decays sufficiently fast as $N\to\infty$; in particular, 
\begin{equation}
    \epsilon_N =O\qty(2^{-yN})= o\qty(1/N^2).
\end{equation}
Therefore, the $O\qty(\qty(\log_2\qty(1/\epsilon_N)+N)^2\sqrt{\epsilon_N}/N)$ term on the left-hand side of~\eqref{eq:main_looser_bound_lower} converges to zero as $N\to\infty$. 
Consequently, by taking the limit of $N\to \infty$, we identify an equivalence relation
\begin{align}
        \label{eq:main_relative_entropy_equivalence}
    &\limsup_{N\to\infty}\min_{\tilde{\sigma}\in\mathcal{F}^{(\epsilon_N)}\qty(\mathcal{H}^{\otimes N})}\qty{\frac{D\left(\rho^{\otimes N}\middle|\middle|\tilde{\sigma}\right)}{N}}\nonumber \\ 
    &=\liminf_{N\to\infty}\min_{\tilde{\sigma}\in\mathcal{F}^{(\epsilon_N)}\qty(\mathcal{H}^{\otimes N})}\qty{\frac{D\left(\rho^{\otimes N}\middle|\middle|\tilde{\sigma}\right)}{N}}\nonumber\\ 
    &=\lim_{N\to\infty}\min_{\sigma\in\mathcal{F}\qty(\mathcal{H}^{\otimes N})}\left\{\frac{D\left(\rho^{\otimes N}\middle|\middle|\sigma\right)}{N}\right\}. 
\end{align}
Applying this relation to~\eqref{eq:main_inequality_3}, we obtain the desired inequality
\begin{equation}
    y \geq \lim_{N\to\infty}\min_{\sigma\in\mathcal{F}\qty(\mathcal{H}^{\otimes N})}\left\{\frac{D\left(\rho^{\otimes N}\middle|\middle|\sigma\right)}{N}\right\}, 
\end{equation}
which completes the proof of the direct part of the generalized quantum Stein's lemma.

Lastly, we clarify our technical contributions in comparison with Ref.~\cite{Brandao2010}. 
To prove~\eqref{eq:main_y_geq_regularized_relative_entropy}, Ref.~\cite{Brandao2010} aimed to show
    \begin{equation}
        \label{eq:prev_statement_1}
\min_{\sigma\in\mathcal{F}\qty(\mathcal{H}^{\otimes N - o(N)})}\qty{\frac{D\left(\rho^{\otimes N - o(N)}\middle|\middle|\sigma\right)}{N - o(N)}} \leq y + o(1), 
\end{equation}
in place of~\eqref{eq:approximately_free_bound} of our proof.
For this purpose, Ref.~\cite{Brandao2010} aimed to develop Lemma III.7 of Ref.~\cite{Brandao2010} in analogy with non-lockability of the relative entropy of entanglement~\cite{Horodecki_locking2005} and showed that Lemma III.7 of Ref.~\cite{Brandao2010} would imply the direct part of the generalized quantum Stein's lemma. 
However, to complete the proof of Lemma III.7, Ref.~\cite{Brandao2010} needed to invoke Lemma III.9 of Ref.~\cite{Brandao2010}, which includes the insoluble logical gap as pointed out by Ref.~\cite{berta2023gap} and thus impedes the proof of~\eqref{eq:prev_statement_1} in the strategy taken in Ref.~\cite{Brandao2010}.
In contrast, we successfully circumvent these obstacles and make a detour by proving an apparently looser bound~\eqref{eq:approximately_free_bound}.
Then, using the continuity bound of the quantum relative entropy in Refs.~\cite{10206734,10129917}, we prove the new relation~\eqref{eq:main_relative_entropy_equivalence} indicating that what we have shown in~\eqref{eq:approximately_free_bound} asymptotically coincides with the desired bound~\eqref{eq:main_y_geq_regularized_relative_entropy}. 

\section*{Data availability}

No data is used in this study.

\section*{Code availability}

No code is used in this study.

\clearpage
\onecolumngrid

\renewcommand{\theequation}{S\arabic{equation}}
\renewcommand{\thetheorem}{S\arabic{theorem}}
\renewcommand{\theproposition}{S\arabic{proposition}}
\renewcommand{\thelemma}{S\arabic{theorem}}
\renewcommand{\thecorollary}{S\arabic{theorem}}
\renewcommand{\thedefinition}{S\arabic{definition}}
\setcounter{equation}{0}
\setcounter{theorem}{0}

\section*{Supplementary Information}

Supplementary Information of ``Generalized Quantum Stein's Lemma: Redeeming Second Law of Resource Theories'' is organized as follows. 
In Sec.~\ref{sec:preliminaries}, we introduce the background materials needed for our main results and proofs. 
In Sec.~\ref{sec:proof of direct part}, we present the details of our proof of the direct part of the generalized quantum Stein's lemma summarized in Methods of the main text, which is the main contribution of our work. 

\section{Preliminaries}
\label{sec:preliminaries}
In this section, we introduce the background materials. 
In Sec.~\ref{subsec:norms}, we define basic concepts, including operators and norms, and introduce our notation. 
In Sec.~\ref{subsec:operator_inequalities}, we list operator inequalities used for our analysis. 
In Sec.~\ref{subsec:entropies}, we summarize definitions of entropic quantities and their properties. 
In Sec.~\ref{subsec:almost_power_states}, we review the concept of almost power states~\cite{Horodecki3_Leung_Oppenheim2008a,Horodecki3_Leung_Oppenheim2008b,renner2006security,renner2007symmetry,Brandao2010}. 
In Sec.~\ref{subsec:QRTs}, we introduce the framework of quantum resource theories (QRTs). 

\subsection{General definitions}
\label{subsec:norms}

As described in the main text, we represent a quantum system by a finite-dimensional complex Hilbert space $\mathcal{H}=\mathbb{C}^d$ for some finite $d$.
The Euclidean norm of a vector $\ket{x} \in \mathcal{H}$ is denoted by $\|\ket{x}\|$. 
The set of quantum states, i.e., positive semidefinite operators with unit trace, of the system $\mathcal{H}$ is denoted by $\mathcal{D}(\mathcal{H})$. 
A vector $\ket{\rho}$ with $\|\ket{\rho}\| = 1$ is called a pure state, which may also be called a state. 
The identity operator is denoted by $\mathds{1}$.
The trace of a linear operator $X$ is denoted by $\Tr[X]$.
A composite system is represented in terms of the tensor product; that is, for any $N\in\mathbb{N}$ with $\mathbb{N}\coloneqq\{1,2,\ldots\}$,  $\mathcal{H}^{\otimes N}$ represents a system composed of $N$ subsystems with  $\mathcal{H}$ representing each subsystem.
For quantum systems $\mathcal{H},\mathcal{H}^E$, and a state $\rho \in \mathcal{D}(\mathcal{H})$, a pure state $\ket{\rho} \in \mathcal{H}\otimes\mathcal{H}^E$ is called a purification of $\rho$ if $\rho = \Tr_{E}\qty[\ket{\rho}\bra{\rho}]$, where we write the partial trace over the auxiliary system $\mathcal{H}^E$ for the purification as $\Tr_{E}$. 
For a state $\rho\in\mathcal{D}\qty(\mathcal{H}^{\otimes N})$ and any $n\in\{1,\ldots,N\}$, the partial trace over the $n$th subsystem is denoted by $\Tr_n\qty[\rho]$.
For a system $\mathcal{H}_1\otimes \mathcal{H}_2\otimes \cdots \otimes\mathcal{H}_N$ composed of a sequence $\{\mathcal{H}_n\}_{n=1}^N$ of $N$ subsystems, we write the partial trace over the first $M$ subsystems as $\Tr_{1,\ldots,M}$.
See also Refs.~\cite{N4,watrous_2018,Wilde_2017} for other conventional notations.

Regarding the asymptotic notation, for $N\in\mathbb{N}$ and functions $f(N)$ and $g(N)$, we write
\begin{equation}
    f(N)=O(g(N))
\end{equation}
if $\limsup_{N\to\infty}\frac{|f(N)|}{g(N)}<\infty$.
Also, we write
\begin{equation}
    f(N)=o(g(N))
\end{equation}
if $\lim_{N\to\infty}\frac{f(N)}{g(N)}=0$.
We write
\begin{equation}
    f(N)=\Theta(g(N))
\end{equation}
if $\limsup_{N\to\infty}\frac{|f(N)|}{g(N)}<\infty$ and $\liminf_{N\to\infty}\frac{f(N)}{g(N)}>0$.

The trace norm of a linear operator $X$ is defined as 
\begin{equation}
    \|X\|_1 \coloneqq \Tr\qty[\sqrt{X^\dagger X}], 
\end{equation}
where $X^\dag$ is the adjoint operator of $X$.
The fidelity between two positive semidefinite operators is defined as
\begin{equation}
\label{eq:fidelity}
F(P,Q)\coloneqq\left\|\sqrt{P}\sqrt{Q}\right\|_1.
\end{equation}
For a Hermitian operator $A$ with spectral decomposition $A=\sum_n \lambda_n\ket{n}\bra{n}$, we write 
\begin{equation}
\qty(A)_+\coloneqq\sum_{n:\lambda_n>0} \lambda_n\ket{n}\bra{n}.
\end{equation}
For a Hermitian operator $A$, we let $\lambda_{\min}(A)$ denote the minimum eigenvalue of $A$.

The set of quantum operations from an input system $\mathcal{H}_\mathrm{in}$ to an output system $\mathcal{H}_\mathrm{out}$ is represented by that of completely positive and trace-preserving (CPTP) linear maps, denoted by $\mathcal{C}\qty(\mathcal{H}_\mathrm{in}\to\mathcal{H}_\mathrm{out})$.
The trace norm has the monotonicity with respect to the application of quantum operations; that is, for any linear operator $X$ and a CPTP linear map $\mathcal{E}$, it holds that $\|\mathcal{E}\qty(X)\|_1 \leq \|X\|_1$. 
Similarly, for a Hermitian operator $A$,  the trace of $\qty(A)_+$ has the following monotonicity with respect to the application of a CPTP linear map $\mathcal{E}$. 
\begin{lemma}[Mononicity of the trace of $\qty(A)_+$]\label{lem:monotonicity}
    For any CPTP linear map $\mathcal{E}$ and any Hermitian operator $A$, we have
    \begin{equation}
        \Tr\qty[\qty(\mathcal{E}\qty(A))_+]\leq\Tr\qty[\qty(A)_+].
    \end{equation}
\end{lemma}
\begin{proof}
    Since $\mathcal{E}$ is a CPTP linear map, it holds that 
    \begin{equation}
        \qty(\mathcal{E}\qty(A))_{+} = \mathcal{E}\qty(\qty(A)_{+}). 
    \end{equation}
    Therefore, we have 
    \begin{align}
        \Tr\qty[\qty(\mathcal{E}\qty(A))_+]
        = \Tr\qty[\mathcal{E}\qty(\qty(A)_{+})]
        = \|\mathcal{E}\qty(\qty(A)_{+})\|_{1} \leq \|\qty(A)_{+}\|_{1} 
        = \Tr\qty[\qty(A)_{+}], 
    \end{align}
    where the inequality follows from the monotonicity of the trace norm. 
\end{proof}

\subsection{Operator inequalities}
\label{subsec:operator_inequalities}
In this section, we review several facts related to operator inequalities. 
First, we see the operator monotonicity of the partial trace. 
\begin{lemma}[Operator monotonicity of partial trace]\label{lem:operator_inequality_partial_trace}
    For any Hermitian operators $A$ and $B$ acting on a composite system $\mathcal{H}_1\otimes\mathcal{H}_2$ and satisfying
    \begin{equation}
        A\geq B,
    \end{equation}
    it holds that
    \begin{equation}
        \Tr_1\qty[A]\geq\Tr_1\qty[B].
    \end{equation}
\end{lemma}

\begin{proof}
    For any orthonormal basis $\{\ket{j}\}_j$ of $\mathcal{H}_1$ and any pure state $\ket{\psi}$ of $\mathcal{H}_2$, it follows from $A\geq B$ that
    \begin{equation}
        \qty(\bra{j}\otimes\bra{\psi})\qty(A-B)\qty(\ket{j}\otimes\ket{\psi})\geq 0.
    \end{equation}
    Therefore, we have
    \begin{equation}
        \bra{\psi}\qty(\Tr_1\qty[A]-\Tr_1\qty[B])\ket{\psi}=\sum_j\qty(\bra{j}\otimes\bra{\psi})\qty(A-B)\qty(\ket{j}\otimes\ket{\psi})\geq 0.
    \end{equation}
\end{proof}

Next, we see that $\log_2$ also has operator monotonicity due to the L\"{o}wner-Heinz Theorem~\cite{carlen2010trace}.  
\begin{lemma}[Theorem~2.6 of Ref.~\cite{carlen2010trace}]\label{lem:operator_monotonicity_log}
    For any positive definite operators $P>0$ and $Q>0$ satisfying
    \begin{equation}
        P\geq Q,
    \end{equation}
    it holds that
    \begin{equation}
        \log_2\qty(P)\geq\log_2\qty(Q).
    \end{equation}
\end{lemma}

Moreover, we state the following useful lemma, which was originally proven in Ref.~\cite{4957652} and repeatedly used in Ref.~\cite{Brandao2010}. 
\begin{lemma}[\label{lem:operator_inequality_fidelity}Lemma C.5 of Ref.~\cite{Brandao2010}, Lemma~5 of Ref.~\cite{4957652}]
For any semidefinite operators $X$, $\Delta$, and state $\rho$ satisfying
\begin{equation}
    \rho\leq X+\Delta,
\end{equation}
with
\begin{equation}
    \Tr\qty[\Delta]<1,
\end{equation}
there exists a state $\tilde{\rho}$ such that
\begin{align}
    &\tilde{\rho}\leq\frac{X}{1-\Tr\qty[\Delta]},\\
    &F\qty(\tilde{\rho},\rho)\geq 1-\Tr\qty[\Delta],
\end{align}
where $F$ is the fidelity in~\eqref{eq:fidelity}.
\end{lemma}

Finally, we show an operator inequality obtained from the trace distance of two states. 
\begin{lemma}[Operator inequality from trace distance]\label{lem:operator_inequality_distance}
    For any states $\rho$ and $\sigma$ satisfying
    \begin{equation}
    0<\|\rho-\sigma\|_1\leq\epsilon,
    \end{equation}
    it holds that
    \begin{equation}
        \rho\leq\sigma+\epsilon\Delta
    \end{equation}
    with a state
    \begin{align}
    \Delta\coloneqq \frac{\qty(\rho-\sigma)_+}{\Tr\qty[\qty(\rho-\sigma)_+]}\geq 0 
    \end{align}
\end{lemma}

\begin{proof}
    For any states $\rho$ and $\sigma$, we have 
    \begin{equation}
        \rho\leq\sigma+\qty(\rho-\sigma)_+. 
    \end{equation}
    Due to $\|\rho-\sigma\|_1 \neq 0$, we have $\qty(\rho-\sigma)_+ \neq 0$, and thus, $\Delta$ is well defined. 
    Then, 
    \begin{align}
        \qty(\rho-\sigma)_+ 
        &= \Tr\qty[\qty(\rho-\sigma)_+] \Delta \\ 
        & \leq \|\rho - \sigma\|_1 \Delta \\ 
        &\leq \epsilon \Delta. 
    \end{align}
    Hence, it holds that 
    \begin{equation}
        \rho\leq\sigma+\epsilon \Delta. 
    \end{equation}
\end{proof}

\subsection{Entropic quantities}
\label{subsec:entropies}

In this section, we summarize definitions of entropic quantities and their properties relevant to our analysis. 
The quantum entropy of a state $\rho$ is defined as 
\begin{equation}
\label{eq:quantum_entropy}
    H(\rho) \coloneqq -\Tr\qty[\rho\log_2(\rho)]. 
\end{equation}
The quantum relative entropy of a state $\rho$ with respect to a state $\sigma$ is defined as
\begin{equation}
\label{eq:relative_entropy}
D\left(\rho\middle|\middle|\sigma\right)\coloneqq\Tr\qty[\rho\qty(\log_2(\rho)-\log_2(\sigma))]
\end{equation}
if the support of $\rho$ is included in the support of $\sigma$, and $\infty$ otherwise.
The binary entropy function $h:[0,1] \to [0,1]$ is defined as
\begin{equation}
\label{eq:binary_entropy}
    h(p)\coloneqq-p\log_2(p)-(1-p)\log_2(1-p), 
\end{equation}
where we write $0\log_2(0) = 0$.  

The quantum entropy is continuous in the sense of the following lemma. 
\begin{lemma}[Lemma~1 of Ref.~\cite{Winter2016}]
\label{lem:asymptotic_continuity_quantum_entorpy}
    For any $\epsilon$ satisfying $0 \leq \epsilon \leq \frac{1}{2}$ and any states $\rho,\sigma \in \mathcal{D}(\mathcal{H})$ of a $d$-dimensional system $\mathcal{H}$ satisfying $\|\rho - \sigma\|_1 \leq \epsilon$, it holds that 
    \begin{equation}
        \qty|H(\rho)-H(\sigma)| \leq 2\epsilon\log_2(d) + h(2\epsilon), 
    \end{equation}
    where $H$ is the quantum entropy in~\eqref{eq:quantum_entropy}, and $h$ is the binary entropy function in~\eqref{eq:binary_entropy}. 
\end{lemma}

Within an appropriately chosen set of states, the quantum relative entropy is continuous with respect to the second argument in the sense of the following lemma. 
\begin{lemma}[Theorem~5.11 of Ref.~\cite{10129917}]
\label{lem:continuity_bound_relative_entropy}
For any state $\rho\in\mathcal{D}(\mathcal{H})$ and any $\tilde{m}\in(0,1)$, we write
\begin{equation}
    \mathcal{S}_{\tilde{m}}\coloneqq\qty{\sigma\in\mathcal{D}(\mathcal{H}):\tilde{m}\rho\leq\sigma}. 
\end{equation}
Let $\epsilon > 0$ be any positive number. 
Then, for any states $\sigma_1,\sigma_2\in\mathcal{S}_{\tilde{m}}$ satisfying
\begin{equation}
    \|\sigma_1-\sigma_2\|_1\leq \epsilon,
\end{equation}
it holds that
\begin{equation}
    \left|D\left(\rho\middle|\middle|\sigma_1\right)-D\left(\rho\middle|\middle|\sigma_2\right)\right|\leq\frac{3\log_2^2\qty(\frac{1}{\tilde{m}})}{1-\tilde{m}}\sqrt{\frac{\epsilon}{2}}.
\end{equation}
\end{lemma}

We show an immediate corollary of Lemma~\ref{lem:continuity_bound_relative_entropy}, which shows a condition for the continuity of the quantum relative entropy in Lemma~\ref{lem:continuity_bound_relative_entropy} in terms of the minimum eigenvalue of the second argument.
\begin{corollary}[Continuity bound of quantum relative entropy with respect to second argument with nonzero minimum eigenvalue]
    \label{cor:continuity_bound_relative_entropy_revised}
    For any $\epsilon > 0$, any state $\rho \in \mathcal{D}(\mathcal{H})$, any full-rank states $\sigma_1, \sigma_2 \in \mathcal{D}(\mathcal{H})$ satisfying
    \begin{equation}
        \|\sigma_1 - \sigma_2\|_1 \leq \epsilon,
    \end{equation} 
    and any $\tilde{m}\in(0,1)$ satisfying
    \begin{equation}
        \label{eq:assumption_m_tilde}
        0 < \tilde{m} \leq \min\{\lambda_{\min}(\sigma_1),\lambda_{\min}(\sigma_2)\}, 
    \end{equation}
    it holds that
    \begin{equation}
        \left|D\left(\rho\middle|\middle|\sigma_1\right)-D\left(\rho\middle|\middle|\sigma_2\right)\right|\leq\frac{3\log_2^2\qty(\frac{1}{\tilde{m}})}{1-\tilde{m}}\sqrt{\frac{\epsilon}{2}}.
    \end{equation}
\end{corollary}
\begin{proof}
    Since $\sigma_1$ and $\sigma_2$ are full-rank states, 
    we obtain from~\eqref{eq:assumption_m_tilde}
    \begin{align}
        &\tilde{m}\rho \leq \tilde{m}\mathds{1} \leq \sigma_1, \\
        &\tilde{m}\rho \leq \tilde{m}\mathds{1} \leq \sigma_2. 
    \end{align}
    Hence, we have $\sigma_1,\sigma_2 \in \mathcal{S}_{\tilde{m}}$ for $\mathcal{S}_{\tilde{m}}$ in Lemma~\ref{lem:continuity_bound_relative_entropy} with this $\tilde{m}$, and using Lemma~\ref{lem:continuity_bound_relative_entropy}, we obtain the desired bound. 
\end{proof}

Lastly, we give an upper bound of the quantum relative entropy in terms of the minimal eigenvalue of the second argument.

\begin{lemma}[\label{lem:relative_entropy_upper_bound}Upper bound of quantum relative entropy]
For any states $\rho,\sigma\in\mathcal{D}\qty(\mathcal{H})$ satisfying $\sigma>0$, i.e., $\lambda_{\min}(\sigma)>0$, we have
\begin{equation}
    D\left(\rho\middle|\middle|\sigma\right)\leq\log_2\qty(\frac{1}{\lambda_{\min}(\sigma)}).
\end{equation}
\end{lemma}

\begin{proof}
Due to
\begin{equation}
    \sigma\geq\lambda_{\min}(\sigma)\mathds{1},
\end{equation}
using Lemma~\ref{lem:operator_monotonicity_log},
we have
\begin{equation}
    \log_2(\sigma)\geq\log_2(\lambda_{\min}(\sigma)\mathds{1}).
\end{equation}
Thus, we obtain
\begin{equation}
    \Tr\qty[\rho\log_2(\sigma)]\geq\Tr\qty[\rho\log_2\qty(\lambda_{\min}(\sigma)\mathds{1})]=\Tr\qty[\rho]\log_2\qty(\lambda_{\min}(\sigma))=\log_2\qty(\lambda_{\min}(\sigma)).
\end{equation}
Therefore, by definition of the quantum relative entropy in~\eqref{eq:relative_entropy}, we have
\begin{align}
   D\left(\rho\middle|\middle|\sigma\right)&=-H\qty(\rho)-\Tr\qty[\rho\log_2\qty(\sigma)]\\
   &\leq -0-\log_2\qty(\lambda_{\min}(\sigma))\\
   &=\log_2\qty(\frac{1}{\lambda_{\min}(\sigma)}),
\end{align}
where we use the fact that the quantum entropy $H(\rho)$ in~\eqref{eq:quantum_entropy} is non-negative.
\end{proof}

\subsection{Almost power states}
\label{subsec:almost_power_states}
In this section, we review the concept of almost power states~\cite{Horodecki3_Leung_Oppenheim2008a,Horodecki3_Leung_Oppenheim2008b,renner2006security,renner2007symmetry,Brandao2010}, which are essential to prove the direct part of the generalized quantum Stein's lemma. 
We show definitions and properties of almost power states that are necessary to understand the main results of our work. 
Readers may also refer to Chapter~4 of Ref.~\cite{renner2006security} for more detailed discussions. 

Let $S_N$ denote the symmetric group of degree $N$.
For $\pi\in S_N$, let $U_\pi$ denote the unitary operator representing the permutation of the state of the $N$ subsystems according to $\pi$.
For any $N\in\mathbb{N}$ and any system $\mathcal{H}$, the symmetric subspace $\Sym\qty(\mathcal{H}^{\otimes N})$ is defined as~\cite{renner2006security}
\begin{equation}
    \Sym\qty(\mathcal{H}^{\otimes N})\coloneqq\qty{\ket{\psi_N}\in\mathcal{H}^{\otimes N}: \forall \pi\in S_N,\,U_\pi\ket{\psi_N}=\ket{\psi_N}}=\spn\qty{\ket{\psi}^{\otimes N}\in\mathcal{H}^{\otimes N}:\ket{\psi}\in\mathcal{H}}.
\end{equation}

For a pure state $\ket{\rho}$ and $R\in\{0,1,\ldots,N\}$, we define the set of states with at least $N-R$ independent and identically distributed (IID) copies of $\ket{\rho}$ as
\begin{equation}
    \mathcal{V}\qty(\mathcal{H}^{\otimes N},\ket{\rho}^{\otimes N-R})\coloneqq\qty{U_\pi\qty(\ket{\rho}^{\otimes N-R}\otimes\ket{\psi_R}): \pi\in S_N,\ket{\psi_R}\in\mathcal{H}^{\otimes R}}.
\end{equation}
The set of almost power states along $\ket{\rho}$ is defined as
\begin{equation}
   \ket{\rho}^{[\otimes,N,R]}\coloneqq\Sym\qty(\mathcal{H}^{\otimes N})\cap\spn\qty(\mathcal{V}\qty(\mathcal{H}^{\otimes N},\ket{\rho}^{\otimes N-R})).
\end{equation}

For any $N,M,R\geq 0$ satisfying $N-M-R\geq 0$, let
\begin{equation}
    \ket{\rho_{N,M,R}}\in\ket{\rho}^{[\otimes,N-M,R]}
\end{equation}
denote an almost power state along $\ket{\rho}$.
By definition, the almost power state along $\ket{\rho}$ can be written as
\begin{equation}
\label{eq:almost_power_state}
    \ket{\rho_{N,M,R}}=\sum_{r=0}^{R}\beta_r\Sym\qty(\ket{\rho}^{\otimes N-M-r}\otimes\ket{\psi_r}),
\end{equation}
where $\ket{\psi_r}\in\qty(\spn\{\ket{\rho}\}^\perp)^{\otimes r}\subset\mathcal{H}^{\otimes r}$ for each $r$ is a permutation-invariant state in the space of the $r$-fold tensor product of the orthogonal complement of $\ket{\rho}$,
\begin{equation}
    \label{eq:defn_sym}
    \Sym\qty(\ket{\rho}^{\otimes N-M-r}\otimes\ket{\psi_r})\coloneqq\frac{1}{\sqrt{\binom{N-M}{r}}}\sum_{\text{$\pi$:$\binom{N-M}{r}$ combinations}}U_\pi\qty(\ket{\rho}^{\otimes N-M-r}\otimes\ket{\psi_r}),
\end{equation}
and 
\begin{equation}
    \sum_{r=0}^{R}|\beta_r|^2=1.
\end{equation}
In~\eqref{eq:defn_sym}, the permutation $\pi$ appearing in the sum corresponds to the choice of $r$ subsystems for $\ket{\psi_r}$ from the $N-M$ subsystems. 

For the analysis of mixed states, it suffices to consider their purification due to the following lemma.
\begin{lemma}[\label{lem:mixed}Lemma III.4 in Ref.~\cite{Brandao2010}]
Consider $N$ IID copies $\rho^{\otimes N}$ of any mixed state $\rho\in\mathcal{D}(\mathcal{H})$, and any permutation-invariant mixed state $\rho_N\in\mathcal{D}(\mathcal{H}^{\otimes N})$, i.e., a state satisfying, for any permutation $\pi\in S_N$,
\begin{equation}
U_\pi \rho_N U_\pi^\dag=\rho_N.
\end{equation}
Then, there exists a purification $\ket{\rho}\in\mathcal{H}\otimes\mathcal{H}^E$ of $\rho$ and a permutation-invariant purification $\ket{\rho_N}\in\Sym\qty(\qty(\mathcal{H}\otimes\mathcal{H}^E)^{\otimes N})$ of $\rho_N$ with $\mathcal{H}^E=\mathcal{H}$ such that
\begin{equation}
    \qty|\braket{\rho_N}{\rho^{\otimes N}}|=F\qty(\rho_N,\rho^{\otimes N}).
\end{equation}
\end{lemma}

Given this purification, the symmetry of the permutation-invariant state and the independence of the almost power state are closely connected by the following lemma, as shown in Ref.~\cite{Brandao2010}.
\begin{lemma}[\label{lem:definitti}Lemma~III.5 in Ref.~\cite{Brandao2010}]
Consider $N$ IID copies $\ket{\rho^{\otimes N}}\coloneqq\ket{\rho}^{\otimes N}$ of any pure state $\ket{\rho}$, and any permutation-invariant state $\ket{\rho_N}\in\Sym\qty(\mathcal{H}^{\otimes N})$ with a nonzero fidelity
\begin{equation}
    \qty|\braket{\rho_N}{\rho^{\otimes N}}|>0.
\end{equation}
Then, for every $M\leq N$, the state
\begin{equation}
    \ket{\rho_{N,M}}\coloneqq\frac{\qty(\bra{\rho^{\otimes M}}\otimes\mathds{1}^{\otimes N-M})\ket{\rho_N}}{\left\|\qty(\bra{\rho^{\otimes M}}\otimes\mathds{1}^{\otimes N-M})\ket{\rho_N}\right\|}\in\mathcal{H}^{\otimes N-M}
\end{equation}
satisfies
\begin{equation}
    \ket{\rho_{N,M}}\bra{\rho_{N,M}}\leq\frac{\Tr_{1,\ldots,M}\qty[\ket{\rho_N}\bra{\rho_N}]}{\qty|\braket{\rho_N}{\rho^{\otimes N}}|^2},
\end{equation}
and for every $R\leq N-M$, there exists an almost power state $\ket{\rho_{N,M,R}}\in\ket{\rho}^{[\otimes,N-M,R]}$ along $\ket{\rho}$ such that
\begin{equation}
   \left\|\ket{\rho_{N,M}}\bra{\rho_{N,M}}-\ket{\rho_{N,M,R}}\bra{\rho_{N,M,R}}\right\|_1\leq\frac{2\sqrt{2}}{\qty|\braket{\rho_N}{\rho^{\otimes N}}|}\mathrm{e}^{-\frac{MR}{2N}}.
\end{equation}
\end{lemma}

\subsection{Quantum resource theories}
\label{subsec:QRTs}
In this section, we introduce the framework of quantum resource theories (QRTs)~\cite{Chitambar2018,Kuroiwa2020}. 
A QRT is specified by choosing, for each pair of input system $\mathcal{H}_\mathrm{in}$ and output system $\mathcal{H}_\mathrm{out}$, a set of free operations as a subset of quantum operations, which we write as $\mathcal{O}\qty(\mathcal{H}_\mathrm{in}\to\mathcal{H}_\mathrm{out})\subseteq\mathcal{C}\qty(\mathcal{H}_\mathrm{in}\to\mathcal{H}_\mathrm{out})$.
A free state is defined as a state that can be generated from any initial state by a free operation; i.e., for any system $\mathcal{H}$, we write the set of free states of $\mathcal{H}$ as $\mathcal{F}(\mathcal{H})\coloneqq\{\sigma\in\mathcal{D}\qty(\mathcal{H}):\forall\mathcal{H}^\prime,\forall\rho\in\mathcal{D}\qty(\mathcal{H}^\prime),\exists\mathcal{E}\in\mathcal{O}\qty(\mathcal{H}^\prime\to\mathcal{H})\text{ such that }\sigma=\mathcal{E}\qty(\rho)\}$.
We work on QRTs with their sets of free states satisfying the following properties.
\begin{enumerate}
    \item \label{p1}For each finite-dimensional system $\mathcal{H}$, $\mathcal{F}(\mathcal{H})$ is closed and convex.
    \item \label{p2}For each finite-dimensional system $\mathcal{H}$, $\mathcal{F}(\mathcal{H})$ contains a full-rank free state $\sigma_{\mathrm{full}}>0$.
    \item \label{p3}For each finite-dimensional system $\mathcal{H}$ and any $N\in\mathbb{N}$, if $\rho\in\mathcal{F}\qty(\mathcal{H}^{\otimes N+1})$, then $\Tr_n\qty[\rho]\in\mathcal{F}\qty(\mathcal{H}^{\otimes N})$ for every $n\in\{1,\ldots,N+1\}$.
    \item \label{p4}For each finite-dimensional system $\mathcal{H}$ and any $N\in\mathbb{N}$, if $\rho\in\mathcal{F}\qty(\mathcal{H})$, then $\rho^{\otimes N}\in\mathcal{F}\qty(\mathcal{H}^{\otimes N})$.
    \item \label{p5}For each finite-dimensional system $\mathcal{H}$ and any $N\in\mathbb{N}$, if $\rho\in\mathcal{F}\qty(\mathcal{H}^{\otimes N})$, then $U_\pi \rho U_\pi^\dag\in\mathcal{F}\qty(\mathcal{H}^{\otimes N})$ for any permutation $\pi\in S_N$, where $S_N$ is the symmetric group of degree $N$, and $U_\pi$ for $\pi\in S_N$ is the unitary operator representing the permutation of the state of the $N$ subsystems of $\mathcal{H}^{\otimes N}$ according to $\pi$.
\end{enumerate}

A non-free state $\rho\in\mathcal{D}(\mathcal{H})\setminus\mathcal{F}(\mathcal{H})$ is called a resource state.
To quantify the amount of resource of a state, we consider a family of real functions $R_\mathcal{H}$ of the state of every system $\mathcal{H}$.
We may omit the subscript $\mathcal{H}$ of $R_\mathcal{H}$ to write $R$ for brevity if obvious from the context.
The function $R$ is called a resource measure if $R$ has a property called monotonicity; i.e., for any free operation $\mathcal{E}$ and any state $\rho$, it should hold that $R(\rho)\geq R(\mathcal{E}(\rho))$.
For example, the relative entropy of resource is defined as~\cite{Chitambar2018,Kuroiwa2020}
\begin{equation}
    R_\mathrm{R}\qty(\rho)\coloneqq\min_{\sigma\in\mathcal{F}\qty(\mathcal{H})}\qty{D\left(\rho\middle|\middle|\sigma\right)},
\end{equation}
where $D\left(\rho\middle|\middle|\sigma\right)$ is the quantum relative entropy in~\eqref{eq:relative_entropy}.
Its variant, the regularized relative entropy of resources, is defined as~\cite{Chitambar2018,Kuroiwa2020}
\begin{align}
    R_\mathrm{R}^\infty\qty(\rho)\coloneqq\lim_{N\to\infty}\frac{R_\mathrm{R}\qty(\rho^{\otimes N})}{N}=\lim_{N\to\infty}\min_{\sigma\in\mathcal{F}\qty(\mathcal{H}^{\otimes N})}\qty{\frac{D\left(\rho^{\otimes N}\middle|\middle|\sigma\right)}{N}},
\end{align}
where the limit exists due to the fact that  $R_\mathrm{R}$ is weakly subadditive~\cite{Chitambar2018,Kuroiwa2020}.

\section{Proof of the direct part of generalized quantum Stein's lemma}
\label{sec:proof of direct part}

In this section, we prove the direct part of the generalized quantum Stein's lemma as explained in Methods of the main text. 
Our main contribution is to construct this proof of the direct part while the original analysis of the direct part of the general quantum Stein's lemma in Ref.~\cite{Brandao2010} has a logical gap as pointed out in Ref.~\cite{berta2023gap}. 
As described in Methods of the main text, our proof involves four steps, and we will explain each step in the subsequent subsections.  
The first step is described in Sec.~\ref{sec:first}, the second step in Sec.~\ref{sec:second}, the third step in Sec.~\ref{sec:third}., and the final step in Sec.~\ref{sec:final}.

In our proof, we will show the contrapositive of Theorem~4 in Methods of the main text, which is stated in the following way. 
\begin{restatable}[Direct part of generalized quantum Stein's lemma (contrapositive)]{theorem}{direct}
\label{thm:direct_part}
For any family $\{\mathcal{F}\qty(\mathcal{H}^{\otimes N})\}_{N\in\mathbb{N}}$ of sets of free states satisfying Properties~\ref{p1}-\ref{p5}, any state $\rho\in\mathcal{D}\qty(\mathcal{H})$, and any $y>0$, if it holds that
\begin{equation}
\label{eq:minimization} \liminf_{N\to\infty}\min_{\sigma\in\mathcal{F}\qty(\mathcal{H}^{\otimes N})}\qty{\Tr\qty[\qty(\rho^{\otimes N}-2^{yN}\sigma)_+]} \in (0,1),
\end{equation}
then we have
\begin{equation}
\label{eq:y_geq}
    y\geq \lim_{N\to\infty}\min_{\sigma\in\mathcal{F}\qty(\mathcal{H}^{\otimes N})}\qty{\frac{D\left(\rho^{\otimes N}\middle|\middle|\sigma\right)}{N}}.
\end{equation}
\end{restatable}

In the following, for each $N\in\mathbb{N}$, we set $M(N)$ and $R(N)$ as any non-negative integers satisfying
\begin{align}
    &N-M(N)\geq 2R(N),\\
    \label{eq:N_M_R1}
    &N-M(N)-R(N)\to\infty,\\
    \label{eq:N_M_R2}
    &\frac{M(N)R(N)}{N}\to\infty,\\
    \label{eq:N_M_R3}
    &\frac{M(N)}{N}\to0,\\
    \label{eq:N_M_R4}
    &\frac{R(N)}{N}\to0,
\end{align}
as $N\to\infty$.
For example, we can choose
\begin{align}
    M(N)&=\Theta\qty(N^{2/3}),\\
    R(N)&=\Theta\qty(N^{2/3}).
\end{align}
For simplicity of notation, we may omit $N$ to write $M(N)$ and $R(N)$ as $M$ and $R$, respectively, if obvious from the context.

\subsection{First step}
\label{sec:first}
In the first step of our proof, in the same way as Ref.~\cite{Brandao2010}, for each $N$, we find a permutation-invariant state $\rho_N \in \mathcal{D}(\mathcal{H}^{\otimes N})$ that has a non-negligible fidelity to $\rho^{\otimes N}$ and is related to an optimal free state $\sigma_N\in\mathcal{F}\qty(\mathcal{H}^{\otimes N})$ in the minimization on the left-hand side of~\eqref{eq:minimization} via an operator inequality. 
In particular, we show the following.

\begin{proposition}[First step]
\label{prp:first_step}
For any family $\{\mathcal{F}\qty(\mathcal{H}^{\otimes N})\}_{N\in\mathbb{N}}$ of sets of free states satisfying Properties~\ref{p1}-\ref{p5}, any state $\rho\in\mathcal{D}\qty(\mathcal{H})$, and any $y>0$, if we have, for a constant $\mu \in (0,1)$,
\begin{equation}
\label{eq:minimization_first_step} \liminf_{N\to\infty}\min_{\sigma\in\mathcal{F}\qty(\mathcal{H}^{\otimes N})}\qty{\Tr\qty[\qty(\rho^{\otimes N}-2^{yN}\sigma)_+]}=1-\mu \in (0,1),
\end{equation}
then there exist a sequence $\{\mu_N>0\}_{N\in\mathbb{N}}$ converging to
\begin{equation}
\label{eq:mu_N_first}
    \lim_{N\to\infty}\mu_N=\mu,
\end{equation}
a sequence $\{\sigma_N\in\mathcal{F}\qty(\mathcal{H}^{\otimes N})\}_{N\in\mathbb{N}}$ of permutation-invariant free states optimally achieving the minimization on the left-hand side of~\eqref{eq:minimization_first_step},
and a sequence $\{\rho_N\in\mathcal{D}\qty(\mathcal{H}^{\otimes N})\}_{N\in\mathbb{N}}$ of permutation-invariant states such that, for each $N$,
\begin{align}
    \label{eq:step1_ineq}
    &\rho_N\leq\frac{2^{yN}}{\mu_N}\sigma_N,\\
    \label{eq:step1_fidelity}
    &F\qty(\rho_N,\rho^{\otimes N})\geq\mu_N,
\end{align}
where $F$ is the fidelity in~\eqref{eq:fidelity}.
\end{proposition}

\begin{proof}
Let $\{\sigma_N\in\mathcal{F}\qty(\mathcal{H}^{\otimes N})\}_{N\in\mathbb{N}}$ be a sequence of optimal solutions in the minimization on the left-hand side of~\eqref{eq:minimization_first_step}, satisfying
\begin{equation}
\label{eq:sigma_N_bound}
    \Tr\qty[\qty(\rho^{\otimes N}-2^{yN}\sigma_N)_+]\geq 1-\mu_N
\end{equation}
for a sequence $\{\mu_N\}_{N\in\mathbb{N}}$ converging to
\begin{equation}
    \lim_{N\to\infty}\mu_N=\mu.
\end{equation}
Due to Lemma~\ref{lem:monotonicity} and Properties~\ref{p1} and~\ref{p5}, we can take $\sigma_N\in\mathcal{F}\qty(\mathcal{H}^{\otimes N})$ to be permutation-invariant.

Applying Lemma~\ref{lem:operator_inequality_fidelity} to
\begin{equation}
    \rho^{\otimes N}\leq 2^{yN}\sigma_N +\qty(\rho^{\otimes N}-2^{yN}\sigma_N)_+,
\end{equation}
we have a sequence $\{\rho_N\}_{N\in\mathbb{N}}$ of states satisfying
\begin{align}
    &\rho_N\leq\frac{2^{yN}}{\mu_N}\sigma_N,\\
    &F\qty(\rho_N,\rho^{\otimes N})\geq\mu_N,
\end{align}
where we use~\eqref{eq:sigma_N_bound}.
Again, due to Lemma~\ref{lem:monotonicity} and the permutation invariance of $\rho^{\otimes N}$ and $\sigma_N$, we can take $\rho_N$ to be permutation-invariant.
\end{proof}

\subsection{Second step}
\label{sec:second}
In the second step, using Proposition~\ref{prp:first_step} shown in the first step, we show an operator inequality between $\rho^{\otimes N}$ and $\sigma_N$, so as to derive a bound of the quantum relative entropy from the operator inequality. 
Although permutation-invariant, $\sigma_N$ and $\rho_N$ defined in the first step may be in non-IID forms, which makes analysis challenging.
To address the non-IIDness, similar to Ref.~\cite{Brandao2010}, we use the almost power states~\cite{Horodecki3_Leung_Oppenheim2008a,Horodecki3_Leung_Oppenheim2008b,renner2006security,renner2007symmetry,Brandao2010} summarized in Sec.~\ref{subsec:almost_power_states} to approximately recover the IID structure of the state. 
For the recovery of IID structure, Ref.~\cite{Brandao2010} represented this approximation by combining an operator inequality and the trace distance, but for our proof, it is crucial to represent this approximation solely in terms of a single operator inequality.
In the following, we will first show an operator inequality to relate the non-IID permutation-invariant state and the almost power state (Proposition~\ref{prp:operator_inequality_non_IID_almost_power}).
Then, we will show an operator inequality to relate the almost power state and the IID state (Proposition~\ref{prp:operator_inequality_almost_power_IID}).
Using the operator inequalities in Propisitions~\ref{prp:operator_inequality_non_IID_almost_power} and~\ref{prp:operator_inequality_almost_power_IID} in combination with that of Proposition~\ref{prp:first_step} in the first step, we derive the operator inequality between $\rho^{\otimes N}$ and $\sigma_N$ (Proposition~\ref{prp:second_step_operator_inequality}).
To convert this operator inequality to a bound of the quantum relative entropy, we also show a formula for this conversion (Proposition~\ref{prp:conversion}).
Finally, applying the conversion in Proposition~\ref{prp:conversion} to the operator inequality in Proposition~\ref{prp:second_step_operator_inequality}, we derive the desired bound of the quantum relative entropy (Proposition~\ref{prp:second_step_relative_entropy}), which achieves the goal of the second step of our proof.

First, we show an operator inequality to relate the non-IID permutation-invariant state and the almost power state, which holds by tracing out a small number of subsystems.

\begin{proposition}[Operator inequality to relate non-IID permutation-invariant state and almost power state]
\label{prp:operator_inequality_non_IID_almost_power}
For any $N\in\mathbb{N}$, any non-negative integers $M,R$ satisfying $M\leq N$ and $R\leq N-M$, any
$\mu_N>0$, any state $\rho\in\mathcal{D}\qty(\mathcal{H})$, and any permutation-invariant state $\rho_N\in\mathcal{D}\qty(\mathcal{H}^{\otimes N})$ satisfying
\begin{equation}
    F\qty(\rho_N,\rho^{\otimes N})\geq\mu_N,
\end{equation}
there exist a purification $\ket{\rho}$ of $\rho$, a permutation-invariant purification $\ket{\rho_N}$ of $\rho_N$, an almost power state along $\ket{\rho}$
\begin{equation}
    \ket{\rho_{N,M,R}}\in\ket{\rho}^{[\otimes,N-M,R]},
\end{equation}
and a state $\Delta_{N,M}$ such that
\begin{equation}
   \ket{\rho_{N,M,R}}\bra{\rho_{N,M,R}}\leq
   \frac{\Tr_{1,\ldots,M}\qty[\ket{\rho_N}\bra{\rho_N}]}{\mu_N^2}+\frac{2\sqrt{2}}{\mu_N}\mathrm{e}^{-\frac{MR}{2N}}\Delta_{N,M}.
\end{equation}
\end{proposition}

\begin{proof}
Due to Lemma~\ref{lem:mixed}, we have a permutation-invariant purification $\ket{\rho_N}$ of $\rho_N$ and a purification $\ket{\rho}$ of $\rho$ satisfying
\begin{equation}
    \qty|\braket{\rho_N}{\rho^{\otimes N}}|\geq\mu_N>0.
\end{equation}
Thus, applying Lemma~\ref{lem:definitti} to $\ket{\rho_N}$ and $\ket{\rho}$, we obtain a state
\begin{equation}
    \ket{\rho_{N,M}}\coloneqq\frac{\qty(\bra{\rho^{\otimes M}}\otimes\mathds{1}^{\otimes N-M})\ket{\rho_N}}{\left\|\qty(\bra{\rho^{\otimes M}}\otimes\mathds{1}^{\otimes N-M})\ket{\rho_N}\right\|}.
\end{equation}
and an almost power state along $\ket{\rho}$
\begin{equation}
    \ket{\rho_{N,M,R}}\in\ket{\rho}^{[\otimes,N-M,R]}
\end{equation}
satisfying
\begin{align}
\label{eq:2}
   &\ket{\rho_{N,M}}\bra{\rho_{N,M}}\leq\frac{\Tr_{1,\ldots,M}\qty[\ket{\rho_N}\bra{\rho_N}]}{\mu_N^2},\\
   \label{eq:2_distance}
   &\left\|\ket{\rho_{N,M}}\bra{\rho_{N,M}}-\ket{\rho_{N,M,R}}\bra{\rho_{N,M,R}}\right\|_1\leq\frac{2\sqrt{2}}{\mu_N}\mathrm{e}^{-\frac{MR}{2N}}.
\end{align}
Thus, using Lemma~\ref{lem:operator_inequality_distance}, we obtain
\begin{align}
\label{eq:3}
   \ket{\rho_{N,M,R}}\bra{\rho_{N,M,R}}&\leq\ket{\rho_{N,M}}\bra{\rho_{N,M}}+\frac{2\sqrt{2}}{\mu_N}\mathrm{e}^{-\frac{MR}{2N}}\Delta_{N,M}\\
   &\leq\frac{\Tr_{1,\ldots,M}\qty[\ket{\rho_N}\bra{\rho_N}]}{\mu_N^2}+\frac{2\sqrt{2}}{\mu_N}\mathrm{e}^{-\frac{MR}{2N}}\Delta_{N,M}
\end{align}
with
\begin{equation}
    \Delta_{N,M}\coloneqq\frac{\qty(\ket{\rho_{N,M,R}}\bra{\rho_{N,M,R}}-\ket{\rho_{N,M}}\bra{\rho_{N,M}})_+}{\Tr\qty[\qty(\ket{\rho_{N,M,R}}\bra{\rho_{N,M,R}}-\ket{\rho_{N,M}}\bra{\rho_{N,M}})_+]}.
\end{equation}
    
\end{proof}

Next, to recover the IID structure from the almost power state, we show an operator inequality to relate the almost power state and the IID state, which holds by tracing out a small number of subsystems.

\begin{proposition}[Operator inequality to relate almost power state and IID state\label{prp:operator_inequality_almost_power_IID}]
    For any $N\in\mathbb{N}$, any non-negative integers $M,R$ satisfying $N-M \geq 2R$, any pure state $\ket{\rho}$, and any almost power state along $\ket{\rho}$
    \begin{equation}
        \ket{\rho_{N,M,R}}\in\ket{\rho}^{[\otimes,N-M,R]},
    \end{equation}
    there exists a state $\Delta_{N,M,R}$ such that
    \begin{equation}
    \label{eq:operator_inequality_non_lockability}
        \ket{\rho}\bra{\rho}^{\otimes N-M-R}\leq 2^{N h\qty(\frac{R}{N-M})}N^2\Tr_{1,\ldots,R}\qty[\ket{\rho_{N,M,R}}\bra{\rho_{N,M,R}}+\frac{2\sqrt{2R}}{N}\Delta_{N,M,R}],
    \end{equation}
    where $h$ is the binary entropy function in~\eqref{eq:binary_entropy}.
\end{proposition}

\begin{proof}
As in the proof of Lemma~III.7 of Ref.~\cite{Brandao2010}, for the almost power state along $\ket{\rho}$ written as~\eqref{eq:almost_power_state}, i.e.,
\begin{equation}
    \ket{\rho_{N,M,R}}=\sum_{r=0}^{R}\beta_r\Sym\qty(\ket{\rho}^{\otimes N-M-r}\otimes\ket{\psi_r})\in\ket{\rho}^{[\otimes,N-M,R]},
\end{equation}
we consider a state $\ket{\tilde{\rho}_{N,M,R}}$ given by
\begin{equation}
    \ket{\tilde{\rho}_{N,M,R}}\coloneqq\frac{\sum_{r:|\beta_r|\geq\frac{1}{N}}\beta_r\Sym\qty(\ket{\rho}^{\otimes N-M-r}\otimes\ket{\psi_r})}{\left\|\sum_{r:|\beta_r|\geq\frac{1}{N}}\beta_r\Sym\qty(\ket{\rho}^{\otimes N-M-r}\otimes\ket{\psi_r})\right\|}.
\end{equation}
To this state, due to $N-M \geq 2R$, we apply Lemma~III.8 of Ref.~\cite{Brandao2010} to obtain
\begin{equation}
\label{eq:operator_inequality_lemma_III_8}
    \ket{\rho}\bra{\rho}^{\otimes N-M-R}\leq 2^{N h\qty(\frac{R}{N-M})}N^2\Tr_{1,\ldots,R}\qty[\ket{\tilde{\rho}_{N,M,R}}\bra{\tilde{\rho}_{N,M,R}}].
\end{equation}
To derive the desired operator inequality~\eqref{eq:operator_inequality_non_lockability} from this operator inequality, we will bound the trace distance between $\ket{\tilde{\rho}_{N,M,R}}\bra{\tilde{\rho}_{N,M,R}}$ and $\ket{\rho_{N,M,R}}\bra{\rho_{N,M,R}}$, which Ref.~\cite{Brandao2010} did not explicitly analyze.

To bound $\left\|\ket{\tilde{\rho}_{N,M,R}}\bra{\tilde{\rho}_{N,M,R}}-\ket{\rho_{N,M,R}}\bra{\rho_{N,M,R}}\right\|_1$, we evaluate
\begin{align}
    &\left\|\sum_{r:|\beta_r|\geq\frac{1}{N}}\beta_r\Sym\qty(\ket{\rho}^{\otimes N-M-r}\otimes\ket{\psi_r})-\ket{\rho_{N,M,R}}\right\|\\
    &=\left\|\sum_{r:|\beta_r|<\frac{1}{N}}\beta_r\Sym\qty(\ket{\rho}^{\otimes N-M-r}\otimes\ket{\psi_r})\right\|\\
    &=\sqrt{\sum_{r:|\beta_r|<\frac{1}{N}}\qty|\beta_r|^2}\\
    &<\sqrt{R\times\frac{1}{N^2}},
\end{align}
where the last inequality follows from the fact that $r$ runs in a subset of $\{1,\ldots,R\}$ with each term satisfying $|\beta_r|^2<1/N^2$.
Since it holds for any non-normalized vectors $\ket{x}$ and $\ket{y}$ that~\cite{Brandao2010}
\begin{align}
    &\left\|\frac{\ket{x}}{\|\ket{x}\|}-\ket{y}\right\|\leq 2\|\ket{x}-\ket{y}\|,\\
    &\left\|\ket{x}\bra{x}-\ket{y}\bra{y}\right\|_1\leq \sqrt{\braket{x}{x}+\braket{y}{y}}\|\ket{x}-\ket{y}\|,
\end{align}
we have
\begin{equation}
\label{eq:4_distance}
    \left\|\ket{\tilde{\rho}_{N,M,R}}\bra{\tilde{\rho}_{N,M,R}}-\ket{\rho_{N,M,R}}\bra{\rho_{N,M,R}}\right\|_1\leq\frac{2\sqrt{2R}}{N}.
\end{equation}

Consequently, 
using Lemma~\ref{lem:operator_inequality_distance}, we obtain from~\eqref{eq:4_distance}
\begin{equation}
\label{eq:5}
   \ket{\tilde{\rho}_{N,M,R}}\bra{\tilde{\rho}_{N,M,R}}\leq\ket{\rho_{N,M,R}}\bra{\rho_{N,M,R}}+\frac{2\sqrt{2R}}{N}\Delta_{N,M,R}
\end{equation}
with
\begin{equation}
    \Delta_{N,M,R}\coloneqq\frac{\qty(\ket{\tilde{\rho}_{N,M,R}}\bra{\tilde{\rho}_{N,M,R}}-\ket{\rho_{N,M,R}}\bra{\rho_{N,M,R}})_+}{\Tr\qty[\qty(\ket{\tilde{\rho}_{N,M,R}}\bra{\tilde{\rho}_{N,M,R}}-\ket{\rho_{N,M,R}}\bra{\rho_{N,M,R}})_+]}.
\end{equation}
Thus, using Lemma~\ref{lem:operator_inequality_partial_trace},  we obtain from~\eqref{eq:operator_inequality_lemma_III_8} and~\eqref{eq:5}
\begin{equation}
        \ket{\rho}\bra{\rho}^{\otimes N-M-R}\leq 2^{N h\qty(\frac{R}{N-M})}N^2\Tr_{1,\ldots,R}\qty[\ket{\rho_{N,M,R}}\bra{\rho_{N,M,R}}+\frac{2\sqrt{2R}}{N}\Delta_{N,M,R}].
\end{equation}
\end{proof}

As a whole, from the operator inequalities of Propositon~\ref{prp:first_step} in the first step and Propositions~\ref{prp:operator_inequality_non_IID_almost_power} and~\ref{prp:operator_inequality_almost_power_IID} in the second step, we obtain the operator inequality shown below.
Note that the state $\tilde{\Delta}_{N-M-R}$ shown in this proposition is denoted by $\tilde{\Delta}_{N-o(N)}$ in Methods of the main text for simplicity of notation.
The definitions of $\epsilon_N$ and $c_N$ appearing in Methods of the main text are also given in this proposition.

\begin{proposition}[Operator inequality summarizing first and second steps]
\label{prp:second_step_operator_inequality}
For any family $\{\mathcal{F}\qty(\mathcal{H}^{\otimes N})\}_{N\in\mathbb{N}}$ of sets of free states satisfying Properties~\ref{p1}-\ref{p5}, any state $\rho\in\mathcal{D}\qty(\mathcal{H})$, any $y>0$, any $N\in\mathbb{N}$, and any non-negative integers $M,R$ satisfying $N-M \geq 2R$, if we have, for a constant $\mu \in (0,1)$,
\begin{equation}
\label{eq:minimization_second_step_operator_inequality} \liminf_{N\to\infty}\min_{\sigma\in\mathcal{F}\qty(\mathcal{H}^{\otimes N})}\qty{\Tr\qty[\qty(\rho^{\otimes N}-2^{yN}\sigma)_+]}=1-\mu \in (0,1),
\end{equation}
then there exist a sequence $\{\mu_N>0\}_{N\in\mathbb{N}}$ converging to
\begin{equation}
\label{eq:mu_N_Second}
    \lim_{N\to\infty}\mu_N=\mu,
\end{equation}
a sequence $\{\sigma_N\in\mathcal{F}\qty(\mathcal{H}^{\otimes N})\}_{N\in\mathbb{N}}$ of permutation-invariant free states optimally achieving the minimization on the left-hand side of~\eqref{eq:minimization_second_step_operator_inequality}, and states $\tilde{\Delta}_{N,M},\tilde{\Delta}_{N,M,R}\in\mathcal{D}\qty(\mathcal{H}^{\otimes N-M-R})$ such that, for
\begin{align}
\label{eq:tilde_sigma_N}
\tilde{\sigma}_{N-M-R}&\coloneqq\frac{\Tr_{1,\ldots,R}\circ\Tr_{1,\ldots,M}\qty[\sigma_N]+\frac{\epsilon_N}{2}\tilde{\Delta}_{N-M-R}}{1+\frac{\epsilon_N}{2}},\\
\label{eq:tilde_Delta}
\tilde{\Delta}_{N-M-R}&\coloneqq\frac{\frac{2\sqrt{2}}{\mu_N\mathrm{e}^{\frac{MR}{2N}}}\tilde{\Delta}_{N,M}+\frac{2\sqrt{2R}}{N}\tilde{\Delta}_{N,M,R}}{\frac{2\sqrt{2}}{\mu_N\mathrm{e}^{\frac{MR}{2N}}}+\frac{2\sqrt{2R}}{N}},\\
\label{eq:epsilon_N}
\epsilon_N&\coloneqq\frac{2\mu_N^3}{2^{yN}}\qty(\frac{2\sqrt{2}}{\mu_N\mathrm{e}^{\frac{MR}{2N}}}+\frac{2\sqrt{2R}}{N}),
\end{align}
it holds that
\begin{align}
\label{eq:opeartor_inequality_first_second}
    \rho^{\otimes N-M-R}\leq
    \frac{2^{N \qty(y+h\qty(\frac{R}{N-M}))}N^2}{\mu_N^3}\qty(1+\frac{\epsilon_N}{2})\tilde{\sigma}_{N-M-R},
\end{align}
where $h$ is the binary entropy function in~\eqref{eq:binary_entropy}.
\end{proposition}

\begin{proof}
Proposition~\ref{prp:first_step} provides the sequences $\{\mu_N>0\}_{N\in\mathbb{N}}$ and $\{\sigma_N\in\mathcal{F}\qty(\mathcal{H}^{\otimes N})\}_{N\in\mathbb{N}}$.
We will show the operator inequality~\eqref{eq:opeartor_inequality_first_second} for these sequences.

We first derive an operator inequality relating $\rho^{\otimes N}$ and $\rho_N$. 
Due to Lemma~\ref{lem:operator_inequality_partial_trace}, from Propositions~\ref{prp:operator_inequality_non_IID_almost_power} and~\ref{prp:operator_inequality_almost_power_IID}, we obtain a chain of operator inequalities
\begin{align}
    \ket{\rho}\bra{\rho}^{\otimes N-M-R}
    &\stackrel{\text{Proposition~\ref{prp:operator_inequality_almost_power_IID}}}{\leq}  2^{N h\qty(\frac{R}{N-M})}N^2\Tr_{1,\ldots,R}\qty[\ket{\rho_{N,M,R}}\bra{\rho_{N,M,R}}+\frac{2\sqrt{2R}}{N}\Delta_{N,M,R}]\\
    &\stackrel{\text{Proposition~\ref{prp:operator_inequality_non_IID_almost_power}}}{\leq}  2^{N h\qty(\frac{R}{N-M})}N^2\Tr_{1,\ldots,R}\qty[\Tr_{1,\ldots,M}\qty[\frac{\ket{\rho_{N}}\bra{\rho_{N}}}{\mu_N^2}]+\qty(\frac{2\sqrt{2}}{\mu_N\mathrm{e}^{\frac{MR}{2N}}}\Delta_{N,M}+\frac{2\sqrt{2R}}{N}\Delta_{N,M,R})]\\
    &= \frac{2^{N h\qty(\frac{R}{N-M})}N^2}{\mu_N^2}\Tr_{1,\ldots,R}\qty[\Tr_{1,\ldots,M}\qty[\ket{\rho_{N}}\bra{\rho_{N}}]+\mu_N^2\qty(\frac{2\sqrt{2}}{\mu_N\mathrm{e}^{\frac{MR}{2N}}}\Delta_{N,M}+\frac{2\sqrt{2R}}{N}\Delta_{N,M,R})],
\end{align}
where $\Delta_{N,M}$ and $\Delta_{N,M,R}$ are states shown in Propositions~\ref{prp:operator_inequality_non_IID_almost_power} and~\ref{prp:operator_inequality_almost_power_IID}, respectively.
Thus, due to Lemma~\ref{lem:operator_inequality_partial_trace},
we obtain
\begin{align}
    \rho^{\otimes N-M-R}&\leq\frac{2^{N h\qty(\frac{R}{N-M})}N^2}{\mu_N^2}\qty[\Tr_{1,\ldots,R}\circ\Tr_{1,\ldots,M}\qty[\rho_{N}]+\mu_N^2\qty(\frac{2\sqrt{2}}{\mu_N}\mathrm{e}^{-\frac{MR}{2N}}\tilde{\Delta}_{N,M}+\frac{2\sqrt{2R}}{N}\tilde{\Delta}_{N,M,R})],
\end{align}
where we define
\begin{align}
\tilde{\Delta}_{N,M,R}&\coloneqq\Tr_E\circ\Tr_{1,\ldots,R}\qty[\Delta_{N,M,R}],\\
\tilde{\Delta}_{N,M}&\coloneqq\Tr_E\circ\Tr_{1,\ldots,R}\qty[\Delta_{N,M}],
\end{align}
and $\Tr_E$ is the partial trace over the auxiliary system used for the purification $\ket{\rho}^{\otimes N-M-R}$ of $\rho^{\otimes N-M-R}$.
This operator inequality can be simplified into
\begin{align}
    \label{eq:6}
    \rho^{\otimes N-M-R}&\leq\frac{2^{N h\qty(\frac{R}{N-M})}N^2}{\mu_N^2}\qty[\Tr_{1,\ldots,R}\circ\Tr_{1,\ldots,M}\qty[\rho_{N}]+c_N\tilde{\Delta}_{N-M-R}],\\
\end{align}
where we use~\eqref{eq:tilde_Delta} and~\eqref{eq:epsilon_N} and write
\begin{equation}
    c_N\coloneqq\frac{\epsilon_N 2^{yN}}{2\mu_N}=O\qty(\epsilon_N 2^{yN}).
\end{equation}
The operator inequality~\eqref{eq:6} serves as the operator inequality relating $\rho^{\otimes N}$ and $\rho_N$. 

Then, we further proceed to obtain an operator inequality relating $\rho^{\otimes N}$ and $\sigma_N$. 
Using Proposition~\ref{prp:first_step},
we obtain from~\eqref{eq:6}
\begin{align}
    \rho^{\otimes N-M-R} &\leq \frac{2^{N h\qty(\frac{R}{N-M})}N^2}{\mu_N^2}\qty[\Tr_{1,\ldots,R}\circ\Tr_{1,\ldots,M}\qty[\frac{2^{yN}}{\mu_N}\sigma_N]+c_N\tilde{\Delta}_{N-M-R}]\\
    &=\frac{2^{N \qty(y+h\qty(\frac{R}{N-M}))}N^2}{\mu_N^3}\qty[\Tr_{1,\ldots,R}\circ\Tr_{1,\ldots,M}\qty[\sigma_N]+\frac{\epsilon_N}{2}\tilde{\Delta}_{N-M-R}]\\
    &=\frac{2^{N \qty(y+h\qty(\frac{R}{N-M}))}N^2}{\mu_N^3}\qty(1+\frac{\epsilon_N}{2})\tilde{\sigma}_{N-M-R},
\end{align}
where we use~\eqref{eq:tilde_sigma_N} in the last line.
\end{proof}

The goal of the second step is to convert this operator inequality to a bound of the quantum relative entropy using the operator monotonicity of $\log_2$ in Lemma~\ref{lem:operator_monotonicity_log}. 
However, an issue here is that if the supports of operators are different, it is not possible to take $\log_2$ of the operators directly at the same time.
To address this issue, we prove the following proposition.
Note that the original analysis of the generalized quantum Stein's lemma in Ref.~\cite{Brandao2010} does not take care of this issue explicitly, and critically to our proof of the general quantum Stein's lemma, this issue also appears in the proof of the strong converse in Ref.~\cite{Brandao2010} (i.e., the proof of Proposition~II.1 in Ref.~\cite{Brandao2010}, which is used in the proof of Corollary~III.3 in Ref.~\cite{Brandao2010});
however, our result here can be used in the analysis of Ref.~\cite{Brandao2010} as well to fix this issue.

\begin{proposition}[Formula for conversion from operator inequality to quantum relative entropy]
\label{prp:conversion}
For any states $\rho,\sigma\in\mathcal{D}\qty(\mathcal{H})$ of any $d$-dimensional system $\mathcal{H}$ and any $\alpha>0$, if it holds that
\begin{equation}
\label{eq:rho_alpha_sigma}
    \rho\leq\alpha\sigma,
\end{equation}
then we have
\begin{equation}
    D\left(\rho\middle|\middle|\sigma\right) \leq \log_2\qty(\alpha).
\end{equation}
\end{proposition}

\begin{proof}
The issue here is that if the supports of $\rho$ and $\sigma$ are different, it is not possible to take $\log_2$ of the operators $\rho$ and $\sigma$ directly at the same time.
To address this issue, we fix an arbitrarily small parameter $\delta$ satisfying $0 < \delta\leq \frac{1}{4}$ and, in place of $\rho$, consider an operator
\begin{equation}
    \frac{\rho+\delta\sigma}{1+\delta},
\end{equation}
which has the same support as $\sigma$ since the support of $\sigma$ satisfying~\eqref{eq:rho_alpha_sigma} always includes the support of $\rho$.

From~\eqref{eq:rho_alpha_sigma}, it follows that
\begin{equation}
    \frac{\rho+\delta\sigma}{1+\delta}\leq
    \frac{\alpha+\delta}{1+\delta}\sigma.
\end{equation}
Then, due to Lemma~\ref{lem:operator_monotonicity_log},
within the support of $\sigma$, it holds that
\begin{equation}
\label{eq:operator2}
    \log_2\qty(\frac{\rho+\delta\sigma}{1+\delta})\leq
    \log_2\qty(\frac{\alpha+\delta}{1+\delta}\sigma).
\end{equation}
Thus, we have
\begin{align}
   \label{eq:limit_delta}
    &\Tr\qty[\frac{\rho+\delta\sigma}{1+\delta}\log_2\qty(\frac{\rho+\delta\sigma}{1+\delta})]
    \leq\Tr\qty[\frac{\rho+\delta\sigma}{1+\delta}
    \log_2\qty(\frac{\alpha+\delta}{1+\delta}\sigma)].
\end{align}
We will take the limit of $\delta\to 0$ to obtain the desired bound of the quantum relative entropy.
On the one hand, due to
\begin{equation}
   \label{eq:limit_delta_1}
   \left\|\frac{\rho+\delta\sigma}{1+\delta}-\rho\right\|_1 \leq\left\|\frac{\delta}{1+\delta}\rho\right\|_1 + \left\|\frac{\delta}{1+\delta}\sigma\right\|_1 \leq 2\delta, 
\end{equation}
using Lemma~\ref{lem:asymptotic_continuity_quantum_entorpy},
we can evaluate the left-hand side of~\eqref{eq:limit_delta} by
\begin{align}
\left|\Tr\qty[\frac{\rho+\delta\sigma}{1+\delta}\log_2\qty(\frac{\rho+\delta\sigma}{1+\delta})]-\Tr\qty[\rho\log_2\qty(\rho)]\right|\leq4\delta\log_2(d)+h(4\delta)
\to 0\quad\text{as $\delta\to 0$},
\end{align}
where $h$ is the binary entropy function in~\eqref{eq:binary_entropy}. 
On the other hand,
we can evaluate the right-hand side of~\eqref{eq:limit_delta} by
\begin{align}
    &\Tr\qty[\frac{\rho + \delta\sigma}{1+\delta}\log_2\qty(\frac{\alpha+\delta}{1+\delta}\sigma)] \nonumber \\
    &=\frac{\Tr\qty[\rho\log_2\qty(\sigma)]+\log_2\qty(\alpha+\delta)-\log_2\qty(1+\delta)}{1+\delta}+\frac{\delta\qty(\Tr\qty[\sigma\log_2\qty(\sigma)]+\log_2\qty(\alpha+\delta)-\log_2\qty(1+\delta))}{1+\delta}\\
    \label{eq:limit_delta_2}
    &\to\Tr\qty[\rho\log_2\qty(\sigma)]+\log_2\qty(\alpha)\quad\text{as $\delta\to 0$}.
\end{align}
As a whole, since the choice of $\delta>0$ can be arbitrarily small, we obtain from~\eqref{eq:limit_delta},~\eqref{eq:limit_delta_1} and~\eqref{eq:limit_delta_2}
\begin{align}
    &\Tr\qty[\rho\log_2\qty(\rho)]\leq \Tr\qty[\rho\log_2\qty(\sigma)]+\log_2\qty(\alpha).
\end{align}
Thus, by definition of the quantum relative entropy in~\eqref{eq:relative_entropy}, we have
\begin{align}
    &D\left(\rho\middle|\middle|\sigma\right)\leq\log_2\qty(\alpha).
\end{align}
    
\end{proof}

Finally, using Proposition~\ref{prp:conversion}, we obtain the following bound of the quantum relative entropy from the operator inequality in Proposition~\ref{prp:second_step_operator_inequality}.

\begin{proposition}[Bound of quantum relative entropy summarizing first and second steps]
\label{prp:second_step_relative_entropy}
For any family $\{\mathcal{F}\qty(\mathcal{H}^{\otimes N})\}_{N\in\mathbb{N}}$ of sets of free states satisfying Properties~\ref{p1}-\ref{p5}, any state $\rho\in\mathcal{D}\qty(\mathcal{H})$, any $y>0$, any $N\in\mathbb{N}$, and any non-negative integers $M,R$ satisfying $N-M \geq 2R$,
if we have, for a constant $\mu \in (0,1)$,
\begin{equation}
\label{eq:minimization_second_step_relative_entropy}
\liminf_{N\to\infty}\min_{\sigma\in\mathcal{F}\qty(\mathcal{H}^{\otimes N})}\qty{\Tr\qty[\qty(\rho^{\otimes N}-2^{yN}\sigma)_+]}=1-\mu \in (0,1),
\end{equation}
then there exist a sequence $\{\mu_N>0\}_{N\in\mathbb{N}}$ converging to
\begin{equation}
\label{eq:mu_N}
    \lim_{N\to\infty}\mu_N=\mu
\end{equation}
and a state $\tilde{\sigma}_{N-M-R}\in\mathcal{D}\qty(\mathcal{H}^{\otimes N-M-R})$ such that
\begin{equation}
    \label{eq:operator_ineq_relative_entropy_exact}
    D\left(\rho^{\otimes N-M-R}\middle|\middle|\tilde{\sigma}_{N-M-R}\right) \leq N\qty(y+h\qty(\frac{R}{N-M}))+\log_2\qty(\frac{N^2}{\mu_N^3})+\log_2\qty(1+\frac{\epsilon_N}{2}),
\end{equation}
where $D$ is the quantum relative entropy in~\eqref{eq:relative_entropy}, $h$ is the binary entropy function in~\eqref{eq:binary_entropy}, and $\epsilon_N$ for each $N$ is given by~\eqref{eq:epsilon_N}.
\end{proposition}

\begin{proof}
Proposition~\ref{prp:second_step_operator_inequality} provides the sequence $\{\mu_N>0\}_{N\in\mathbb{N}}$ and the state $\tilde{\sigma}_{N-M-R}\in\mathcal{D}\qty(\mathcal{H}^{\otimes N-M-R})$.
We will show~\eqref{eq:operator_ineq_relative_entropy_exact} for this choice.

Proposition~\ref{prp:second_step_operator_inequality} yields
\begin{equation}
\label{eq:operators}
    \rho^{\otimes N-M-R}\leq
    \frac{2^{N \qty(y+h\qty(\frac{R}{N-M}))}N^2}{\mu_N^3}\qty(1+\frac{\epsilon_N}{2})\tilde{\sigma}_{N-M-R}. 
\end{equation}
Using Proposition~\ref{prp:conversion}, we obtain from~\eqref{eq:operators}
\begin{equation}
    D\left(\rho^{\otimes N-M-R}\middle|\middle|\tilde{\sigma}_{N-M-R}\right) \leq \log_2\qty(\frac{2^{N \qty(y+h\qty(\frac{R}{N-M}))}N^2}{\mu_N^3}\qty(1+\frac{\epsilon_N}{2})). 
\end{equation}
Therefore, due to
\begin{equation}
    \log_2\qty(\frac{2^{N \qty(y+h\qty(\frac{R}{N-M}))}N^2}{\mu_N^3}\qty(1+\frac{\epsilon_N}{2}))=N\qty(y+h\qty(\frac{R}{N-M}))+\log_2\qty(\frac{N^2}{\mu_N^3})+\log_2\qty(1+\frac{\epsilon_N}{2}),
\end{equation}
we obtain
\begin{equation}
    D\left(\rho^{\otimes N-M-R}\middle|\middle|\tilde{\sigma}_{N-M-R}\right) \leq N\qty(y+h\qty(\frac{R}{N-M}))+\log_2\qty(\frac{N^2}{\mu_N^3})+\log_2\qty(1+\frac{\epsilon_N}{2}). 
\end{equation}
\end{proof}

\subsection{Third step}
\label{sec:third}
In the third step, we develop a technique for properly dealing with the difference between the quantum relative entropies with respect to the optimal free state and the state $\tilde{\sigma}_{N-M-R}$ in Proposition~\ref{prp:second_step_relative_entropy} obtained from our approximation, which does not appear in the analysis of Ref.~\cite{Brandao2010}.
For this purpose, we introduce a concept of asymptotically free states (Definition~\ref{def:asym_free_state}), which are states close to some free state up to an asymptotically vanishing error in terms of trace distance.
We then regard $\tilde{\sigma}_{N-M-R}$ as an asymptotically free state with respect to an appropriately chosen error parameter (Proposition~\ref{prp:asymptotically_free_state}).
Finally, we analyze the regularization of the quantum relative entropy in Proposition~\ref{prp:second_step_relative_entropy} and represent it in terms of the set of asymptotically free states shown in Proposition~\ref{prp:asymptotically_free_state} (Proposition~\ref{prp:third_step}), which is the goal of the third step.

We define asymptotically free states as follows.

\begin{definition}[Asymptotically free states]
\label{def:asym_free_state}
For any family $\{\mathcal{F}\qty(\mathcal{H}^{\otimes N})\}_{N\in\mathbb{N}}$ of sets of free states and any sequence $\{\epsilon_N>0\}_{N\in\mathbb{N}}$ of asymptotically vanishing error parameters $\lim_{N\to\infty}\epsilon_N=0$, the family $\{\mathcal{F}^{(\epsilon_N)}\qty(\mathcal{H}^{\otimes N})\}_{N\in\mathbb{N}}$ of sets of asymptotically free states with respect to $\{\epsilon_N>0\}_{N\in\mathbb{N}}$ is defined as
\begin{align}
\label{eq:asym_free_state_supple}
\mathcal{F}^{(\epsilon_N)}\qty(\mathcal{H}^{\otimes N})\coloneqq\qty{\tilde{\sigma}\in\mathcal{D}\qty(\mathcal{H}^{\otimes N}):\exists\sigma\in\mathcal{F}\qty(\mathcal{H}^{\otimes N}),\,\|\tilde{\sigma}-\sigma\|_1\leq\epsilon_N}.
\end{align}
A state $\tilde{\sigma}\in\mathcal{F}^{(\epsilon_N)}\qty(\mathcal{H}^{\otimes N})$ in a set of such a family is called an asymptotically free state.
\end{definition}

With Definition~\ref{def:asym_free_state}, we identify the set of asymptotically free states including the state $\tilde{\sigma}_{N-M-R}$ in~\eqref{eq:tilde_sigma_N} and also appearing in Proposition~\ref{prp:second_step_relative_entropy} as follows.

\begin{proposition}[State in the set of asymptotically free states]
\label{prp:asymptotically_free_state}
For any family $\{\mathcal{F}\qty(\mathcal{H}^{\otimes N})\}_{N\in\mathbb{N}}$ of sets of free states satisfying Properties~\ref{p1}-\ref{p5}, any $N\in\mathbb{N}$, any non-negative integers $M,R$ satisfying $N-M-R\geq 0$,
any free state $\sigma_N\in\mathcal{F}\qty(\mathcal{H}^{\otimes N})$, any state $\tilde{\Delta}_{N-M-R}\in\mathcal{D}\qty(\mathcal{H}^{\otimes N-M-R})$, and any $\epsilon_N>0$,
let $\tilde{\sigma}_{N-M-R}$ be a state
\begin{equation}
\tilde{\sigma}_{N-M-R}\coloneqq\frac{\Tr_{1,\ldots,R}\circ\Tr_{1,\ldots,M}\qty[\sigma_N]+\frac{\epsilon_N}{2}\tilde{\Delta}_{N-M-R}}{1+\frac{\epsilon_N}{2}}.
\end{equation}
Then, $\tilde{\sigma}_{N-M-R}$ is in the set $\mathcal{F}^{(\epsilon_N)}\qty(\mathcal{H}^{\otimes N-M-R})$ of asymptotically free states, i.e.,
\begin{equation}
\tilde{\sigma}_{N-M-R}\in\mathcal{F}^{(\epsilon_N)}\qty(\mathcal{H}^{\otimes N-M-R}).
\end{equation}
\end{proposition}
\begin{proof}
Given any free state $\sigma_N\in\mathcal{F}\qty(\mathcal{H}^{\otimes N})$, after tracing out $M+R$ subsystems from $\sigma_N$, the resulting state is also free  due to Property~\ref{p3}; i.e., we have
\begin{equation}
    \Tr_{1,\ldots,R}\circ\Tr_{1,\ldots,M}\qty[\sigma_N]\in\mathcal{F}\qty(\mathcal{H}^{\otimes N-M-R}). 
\end{equation}
By definition of $\tilde{\sigma}_{N-M-R}$, we have
\begin{equation}
    \left\|\tilde{\sigma}_{N-M-R}-\Tr_{1,\ldots,R}\circ\Tr_{1,\ldots,M}\qty[\sigma_N]\right\|_1\leq
    \left\|\frac{\frac{\epsilon_N}{2}}{1 + \frac{\epsilon_N}{2}}\Tr_{1,\ldots,R}\circ\Tr_{1,\ldots,M}\qty[\sigma_N]\right\|_1+\left\|\frac{\frac{\epsilon_N}{2}}{1 + \frac{\epsilon_N}{2}}\tilde{\Delta}_{N-M-R}\right\|_1\leq\epsilon_N, 
\end{equation}
and thus, it holds that
\begin{equation}
    \tilde{\sigma}_{N-M-R}\in\mathcal{F}^{(\epsilon_N)}\qty(\mathcal{H}^{\otimes N-M-R}). 
\end{equation}
    
\end{proof}

Using the set of asymptotically free states in Proposition~\ref{prp:asymptotically_free_state},
we characterize the regularization of the bound of quantum relative entropy obtained in Proposition~\ref{prp:second_step_relative_entropy} as follows. 

\begin{proposition}[Characterization of regularization of quantum relative entropy in terms of asymptotically free states summarizing first, second, and third steps]
\label{prp:third_step}
For any family $\{\mathcal{F}\qty(\mathcal{H}^{\otimes N})\}_{N\in\mathbb{N}}$ of sets of free states satisfying Properties~\ref{p1}-\ref{p5}, any state $\rho\in\mathcal{D}\qty(\mathcal{H})$, any $y>0$, any $N\in\mathbb{N}$, and any non-negative integers $M(N),R(N)$ satisfying
\begin{align}
    &N-M(N) \geq 2R(N),\\
    \label{eq:N_M_R1_prp}
    &N-M(N)-R(N)\to\infty,\\
    \label{eq:N_M_R2_prp}
    &\frac{M(N)R(N)}{N}\to\infty,\\
    \label{eq:N_M_R3_prp}
    &\frac{M(N)}{N}\to0,\\
    \label{eq:N_M_R4_prp}
    &\frac{R(N)}{N}\to0,
\end{align}
if we have, for a constant $\mu \in (0,1)$,
\begin{equation}
\label{eq:minimization_third_step}
\liminf_{N\to\infty}\min_{\sigma\in\mathcal{F}\qty(\mathcal{H}^{\otimes N})}\qty{\Tr\qty[\qty(\rho^{\otimes N}-2^{yN}\sigma)_+]}=1-\mu \in (0,1),
\end{equation}
then there exists a sequence $\{\mu_N>0\}_{N\in\mathbb{N}}$ converging to
\begin{equation}
\label{eq:mu_N_prp}
    \lim_{N\to\infty}\mu_N=\mu
\end{equation}
such that
\begin{equation}
\label{eq:y_geq_variant}
    y\geq \limsup_{N\to\infty}\min_{\tilde{\sigma}\in\mathcal{F}^{(\epsilon_N)}\qty(\mathcal{H}^{\otimes N})}\qty{\frac{D\left(\rho^{\otimes N-M-R}\middle|\middle|\sigma\right)}{N-M-R}}, 
\end{equation}
where $\mathcal{F}^{(\epsilon_N)}\qty(\mathcal{H}^{\otimes N})$ is the set of asymptotically free states given by~\eqref{eq:asym_free_state_supple} with respect to $\{\epsilon_N > 0\}_{N\in\mathbb{N}}$ given by~\eqref{eq:epsilon_N} depending on $N$, $M$, $R$, $y$, and $\mu_N$. 
\end{proposition}

\begin{proof}
Proposition~\ref{prp:second_step_relative_entropy} provides the sequence $\{\mu_N>0\}_{N\in\mathbb{N}}$ and the state $\tilde{\sigma}_{N-M-R}\in\mathcal{D}\qty(\mathcal{H}^{\otimes N-M-R})$.
For $\{\epsilon_N > 0\}_{N\in\mathbb{N}}$ given by~\eqref{eq:epsilon_N}, i.e.,
\begin{equation}
\epsilon_N=\frac{2\mu_N^3}{2^{yN}}\qty(\frac{2\sqrt{2}}{\mu_N\mathrm{e}^{\frac{MR}{2N}}}+\frac{2\sqrt{2R}}{N}),
\end{equation}
under our assumption of
\begin{equation}
    y>0,
\end{equation}
and under~\eqref{eq:N_M_R2_prp},~\eqref{eq:N_M_R4_prp}, and~\eqref{eq:mu_N_prp},
it holds that
\begin{equation}
    \lim_{N\to\infty}\epsilon_N=0.
\end{equation}
Thus, due to Proposition~\ref{prp:asymptotically_free_state},  $\tilde{\sigma}_{N-M-R}\in\mathcal{D}\qty(\mathcal{H}^{\otimes N-M-R})$ given by~\eqref{eq:tilde_sigma_N} as shown in Proposition~\ref{prp:second_step_relative_entropy} is an asymptotically free state with respect to $\{\epsilon_N>0\}_{N\in\mathbb{N}}$, i.e.,
\begin{equation}
    \label{eq:tilde_sigma_asym_free}
    \tilde{\sigma}_{N-M-R}\in\mathcal{F}^{(\epsilon_N)}\qty(\mathcal{H}^{\otimes N}).
\end{equation}

Therefore, due to Proposition~\ref{prp:second_step_relative_entropy}, we have
\begin{align}
&\min_{\tilde{\sigma}\in\mathcal{F}^{(\epsilon_N)}(\mathcal{H}^{\otimes N-M-R})}\qty{\frac{D\left(\rho^{\otimes N-M-R}\middle|\middle|\tilde{\sigma}\right)}{N-M-R}}\nonumber\\
&\stackrel{\eqref{eq:tilde_sigma_asym_free}}{\leq}\frac{D\left(\rho^{\otimes N-M-R}\middle|\middle|\tilde{\sigma}_{N-M-R}\right)}{N-M-R}\\
&\stackrel{\text{Proposition~\ref{prp:second_step_relative_entropy}}}{\leq}\frac{1}{N-M-R}\qty(N\qty(y+h\qty(\frac{R}{N-M}))+\log_2\qty(\frac{N^2}{\mu_N^3})+\log_2\qty(1+\frac{\epsilon_N}{2}))\\
\label{eq:limit_y}
&\to y\quad\text{as $N\to\infty$},
\end{align}
where the last line follows from the choice of $M$ and $R$ specified in~\eqref{eq:N_M_R1_prp},~\eqref{eq:N_M_R2_prp},~\eqref{eq:N_M_R3_prp}, and~\eqref{eq:N_M_R4_prp}. 
Therefore, we obtain
\begin{equation}
    y\geq\limsup_{N\to\infty}\min_{\tilde{\sigma}\in\mathcal{F}^{(\epsilon_N)}(\mathcal{H}^{\otimes N-M-R})}\qty{\frac{D\left(\rho^{\otimes N-M-R}\middle|\middle|\tilde{\sigma}\right)}{N-M-R}}.  
\end{equation}
\end{proof}

\subsection{Final step}
\label{sec:final}
In the final step, we develop a technique to address the difference between the regularized relative entropy of resource appearing in Theorem~\ref{thm:direct_part} to be proven and the corresponding quantity appearing in Proposition~\ref{prp:third_step}.
Proposition~\ref{prp:third_step} almost shows the desired bound in Theorem~\ref{thm:direct_part}; however, the minimization on the right-hand side of~\eqref{eq:y_geq_variant} in Proposition~\ref{prp:third_step} is taken over the set $\mathcal{F}^{(\epsilon_N)}(\mathcal{H}^{\otimes N})$ of asymptotically free states while the minimization in the regularized relative entropy of resource in~\eqref{eq:y_geq} of Theorem~\ref{thm:direct_part} is taken over the set $\mathcal{F}\qty(\mathcal{H}^{\otimes N})$ of free states. 
We will study this difference using the continuity bound of the quantum relative entropy in  Corollary~\ref{cor:continuity_bound_relative_entropy_revised}, so as to show an equivalence relation indicating that the difference vanishes after the regularization (Proposition~\ref{lem:limit_regularization}).
Finally, applying Proposition~\ref{lem:limit_regularization} to Proposition~\ref{prp:third_step}, we prove Theorem~\ref{thm:direct_part}.

In particular, we prove the following equivalence relation.

\begin{proposition}[Equivalence relation between regularized relative entropies of resource with respect to the sets of free states and asymptotically free states]
\label{lem:limit_regularization}
For any state $\rho\in\mathcal{D}(\mathcal{H})$ and any sequence $\{\epsilon_N>0\}_{N\in\mathbb{N}}$ satisfying
\begin{equation}
  \epsilon_N = o\qty(1/N^2)\quad\text{as $N\to\infty$},  
\end{equation}
we have an equivalence relation
\begin{equation}
\label{eq:minimization_limit_regularization}
    \limsup_{N\to\infty}\min_{\tilde{\sigma}\in\mathcal{F}^{(\epsilon_N)}\qty(\mathcal{H}^{\otimes N})}\qty{\frac{D\left(\rho^{\otimes N}\middle|\middle|\tilde{\sigma}\right)}{N}} = \liminf_{N\to\infty}\min_{\tilde{\sigma}\in\mathcal{F}^{(\epsilon_N)}\qty(\mathcal{H}^{\otimes N})}\qty{\frac{D\left(\rho^{\otimes N}\middle|\middle|\tilde{\sigma}\right)}{N}}=\lim_{N\to\infty}\min_{\sigma\in\mathcal{F}(\mathcal{H}^{\otimes N})}\qty{\frac{D\left(\rho^{\otimes N}\middle|\middle|\sigma\right)}{N}},
\end{equation}
where $\mathcal{F}\qty(\mathcal{H}^{\otimes N})$ is the set of free states satisfying Properties~\ref{p1}-\ref{p5}, and $\mathcal{F}^{(\epsilon_N)}\qty(\mathcal{H}^{\otimes N})$ is the set of asymptotically free states in Definition~\ref{def:asym_free_state}.
\end{proposition}

\begin{proof}
    Let $\{\tilde{\sigma}_N\in\mathcal{F}^{(\epsilon_N)}\qty(\mathcal{H}^{\otimes N})\}_{N\in\mathbb{N}}$ and  $\{\sigma_N\in\mathcal{F}\qty(\mathcal{H}^{\otimes N})\}_{N\in\mathbb{N}}$ denote optimal sequences in the minimization of~\eqref{eq:minimization_limit_regularization}.
    Due to $\mathcal{F}^{(\epsilon_N)}\qty(\mathcal{H}^{\otimes N})\supseteq\mathcal{F}\qty(\mathcal{H}^{\otimes N})$, we have
    \begin{equation}
    \label{eq:upper_sigma_N}
        D\left(\rho^{\otimes N}\middle|\middle|\tilde{\sigma}_N\right)\leq D\left(\rho^{\otimes N}\middle|\middle|\sigma_N\right).
    \end{equation}
    In the case of $\rho\in\mathcal{F}\qty(\mathcal{H})$, it always holds that $D\left(\rho^{\otimes N}\middle|\middle|\tilde{\sigma}_N\right)=D\left(\rho^{\otimes N}\middle|\middle|\sigma_N\right)=0$ due to Property~\ref{p4}; thus, the equivalence relation~\eqref{eq:minimization_limit_regularization} holds trivially.
    In the following, we will consider the other case, i.e.,
    \begin{equation}
    \label{eq:rho_non_free}
    \rho\not\in\mathcal{F}\qty(\mathcal{H});
    \end{equation}
    in this case, due to Property~\ref{p3}, it also holds that
    $\rho^{\otimes N}\not\in\mathcal{F}\qty(\mathcal{H}^{\otimes N})$.

    For each $N$, due to $\tilde{\sigma}_N\in\mathcal{F}^{(\epsilon_N)}\qty(\mathcal{H}^{\otimes N})$, there exists a free state $\sigma\in\mathcal{F}\qty(\mathcal{H}^{\otimes N})$ such that 
    \begin{equation}
    \label{eq:distance_upper_bound}
        \|\tilde{\sigma}_N-\sigma\|_1\leq \epsilon_N.
    \end{equation}
    Under the condition~\eqref{eq:rho_non_free},
    we will show
    \begin{equation}
    \label{eq:distance_lower_bound}
        \tilde{\sigma}_N\neq \sigma.
    \end{equation}
    To this goal, we will show that there exists an asymptotically free state
    \begin{equation}
    \label{eq:tilde_sigma_asymptotically_free}
        \tilde{\sigma} \in \mathcal{F}^{(\epsilon_N)}\qty(\mathcal{H}^{\otimes N})
    \end{equation}
    such that 
    \begin{equation}
        \label{eq:relative_entropy_strict_ineq}
        D\left(\rho^{\otimes N}\middle|\middle|\tilde{\sigma}\right) < D\left(\rho^{\otimes N}\middle|\middle|\sigma_N\right), 
    \end{equation} 
    which implies
    \begin{equation}
        D\left(\rho^{\otimes N}\middle|\middle|\tilde{\sigma}_N\right) \leq D\left(\rho^{\otimes N}\middle|\middle|\tilde{\sigma}\right) < D\left(\rho^{\otimes N}\middle|\middle|\sigma_N\right) \leq D\left(\rho^{\otimes N}\middle|\middle|\sigma\right), 
    \end{equation}
    and thus~\eqref{eq:tilde_sigma_asymptotically_free}.
    We here construct $\tilde{\sigma}$ as
    \begin{equation}
    \label{eq:tilde_sigma_construction}
        \tilde{\sigma} \coloneqq \frac{\sigma_N + \frac{\epsilon_N}{2}\rho^{\otimes N}}{1 + \frac{\epsilon_N}{2}}. 
    \end{equation}
    For this $\tilde{\sigma}$, we have~\eqref{eq:tilde_sigma_asymptotically_free} since it holds that
    \begin{equation}
        \left\|\tilde{\sigma} - \sigma_N \right\|_1 \leq \left\|\frac{\frac{\epsilon_N}{2}}{1 + \frac{\epsilon_N}{2}} \sigma_N\right\|_1 + \left\|\frac{\frac{\epsilon_N}{2}}{1 + \frac{\epsilon_N}{2}} \rho^{\otimes N} \right\|_1 \leq \epsilon_N,
    \end{equation}
    and $\sigma_N\in\mathcal{F}\qty(\mathcal{H}^{\otimes N})$.
    It remains to show that $\tilde{\sigma}$ in~\eqref{eq:tilde_sigma_construction} satisfies~\eqref{eq:relative_entropy_strict_ineq}. 
    To show this, using the joint convexity of the quantum relative entropy~\cite{N4}, we obtain
    \begin{align}
        \label{eq:relative_entropy_tilde_sigma_1}
        D\left(\rho^{\otimes N}\middle|\middle|\tilde{\sigma}\right)
        = D\left(\rho^{\otimes N}\middle|\middle|\frac{\sigma_N + \frac{\epsilon_N}{2}\rho^{\otimes N}}{1 + \frac{\epsilon_N}{2}}\right) \leq 
        \frac{1}{1 + \frac{\epsilon_N}{2}} D\left(\rho^{\otimes N}\middle|\middle|\sigma_N\right) +  \frac{\frac{\epsilon_N}{2}}{1 + \frac{\epsilon_N}{2}} D\left(\rho^{\otimes N}\middle|\middle|\rho^{\otimes N}\right).  
    \end{align}
    On the one hand, 
    the first term of the right-hand side of~\eqref{eq:relative_entropy_tilde_sigma_1} can be bounded by
     \begin{equation}
        \label{eq:relative_entropy_tilde_sigma_2}
        \frac{1}{1 + \frac{\epsilon_N}{2}} D\left(\rho^{\otimes N}\middle|\middle|\sigma_N\right) < D\left(\rho^{\otimes N}\middle|\middle|\sigma_N\right),
    \end{equation}
    where we use $\epsilon_N > 0$ and $D\left(\rho^{\otimes N}\middle|\middle|\sigma_N\right) > 0$ due to~\eqref{eq:rho_non_free}; 
    on the other hand, the second term of the right-hand side of~\eqref{eq:relative_entropy_tilde_sigma_1} is
    \begin{equation}
        \label{eq:relative_entropy_tilde_sigma_3}
        \frac{\frac{\epsilon_N}{2}}{1 + \frac{\epsilon_N}{2}}D\left(\rho^{\otimes N}\middle|\middle|\rho^{\otimes N}\right) = 0
    \end{equation}
    by definition of the quantum relative entropy in~\eqref{eq:relative_entropy}. 
    Combining~\eqref{eq:relative_entropy_tilde_sigma_1},~\eqref{eq:relative_entropy_tilde_sigma_2}, and~\eqref{eq:relative_entropy_tilde_sigma_3}, we obtain~\eqref{eq:relative_entropy_strict_ineq}, which yields~\eqref{eq:distance_lower_bound}.
    As a whole, under the condition~\eqref{eq:rho_non_free}, it follows from~\eqref{eq:distance_upper_bound} and~\eqref{eq:distance_lower_bound} that
    \begin{equation}
    \label{eq:distance_bound}
        0<\|\tilde{\sigma}_N-\sigma\|_1\leq \epsilon_N.
    \end{equation}

    Due to~\eqref{eq:distance_bound}, using Lemma~\ref{lem:operator_inequality_distance},
    we obtain an operator inequality
    \begin{equation}
    \label{eq:inequality_tilde_sigma_sigma_epsilon_N_Delta}
        \tilde{\sigma}_N \leq \sigma + \epsilon_N \Delta
    \end{equation}
    with a state
    \begin{align}
        \Delta &\coloneqq \frac{\qty(\tilde{\sigma}_N - \sigma)_+}{\Tr\qty[\qty(\tilde{\sigma}_N - \sigma)_+]}.
    \end{align}
    Then, applying a similar argument to the proof of Lemma~2 in Ref.~\cite{Winter2016} to this operator inequality,
    we obtain a state 
    \begin{equation}
        \omega\coloneqq\frac{1}{1+\epsilon_N}\sigma+\frac{\epsilon_N}{1+\epsilon_N}\Delta=\frac{1}{1+\epsilon_N}\tilde{\sigma}_N+\frac{\epsilon_N}{1+\epsilon_N}\tilde{\Delta} 
    \end{equation}
    with a state $\tilde{\Delta}$ defined as 
    \begin{align}
        \tilde{\Delta} &\coloneqq \Delta + \frac{1}{\epsilon_N}\qty(\sigma - \tilde{\sigma}_N). 
    \end{align}
    Note that $\tilde{\Delta}$ is a state since
    \begin{align}
        \Tr\qty[\tilde{\Delta}]&=\Tr\qty[\Delta]+\frac{1}{\epsilon_N}\Tr\qty[\sigma-\tilde{\sigma}_N]=1,\\
        \tilde{\Delta}&=\frac{\epsilon_N\Delta+\sigma - \tilde{\sigma}_N}{\epsilon_N}\geq 0,
    \end{align}
    where the last operator inequality follows from~\eqref{eq:inequality_tilde_sigma_sigma_epsilon_N_Delta}. 
    We further introduce
    \begin{align}
        \omega^\prime&\coloneqq\frac{1}{1+2\epsilon_N}\sigma+\frac{\epsilon_N}{1+2\epsilon_N}\Delta+\frac{\epsilon_N}{1+2\epsilon_N}\sigma_{\mathrm{full}}^{\otimes N}=\frac{1}{1+2\epsilon_N}\tilde{\sigma}_N+\frac{\epsilon_N}{1+2\epsilon_N}\tilde{\Delta}+\frac{\epsilon_N}{1+2\epsilon_N}\sigma_{\mathrm{full}}^{\otimes N},\\
        \sigma^\prime&\coloneqq\frac{1}{1+2\epsilon_N}\sigma+\frac{2\epsilon_N}{1+2\epsilon_N}\sigma_{\mathrm{full}}^{\otimes N}\in\mathcal{F}\qty(\mathcal{H}^{\otimes N}), 
    \end{align}
    where $\sigma_{\mathrm{full}}\in\mathcal{F}\qty(\mathcal{H})$ is a full-rank free state due to Property~\ref{p2} with the nonzero minimum eigenvalue
    \begin{equation}
    \label{eq:lambda_min_condition}
        \lambda_{\min}(\sigma_{\mathrm{full}})>0,
    \end{equation}
    $\sigma_{\mathrm{full}}^{\otimes N}\in\mathcal{F}\qty(\mathcal{H}^{\otimes N})$ due to Property~\ref{p4}, and $\sigma^\prime$ is a free state due to Property~\ref{p1}.
    Due to~\eqref{eq:lambda_min_condition}, we fix a positive constant $\lambda_{\min}\in(0,1)$ satisfying
    \begin{equation}
    \lambda_{\min}\leq\lambda_{\min}(\sigma_{\mathrm{full}}).
    \end{equation}
    By definition, we have
    \begin{equation}
    \label{eq:distance_omega_sigma}
        \|\omega^\prime-\sigma^\prime\|_1\leq\left\|\frac{\epsilon_N}{1+2\epsilon_N}\Delta\right\|_1+\left\|\frac{\epsilon_N}{1+2\epsilon_N}\sigma_{\mathrm{full}}^{\otimes N}\right\|_1 \leq2\epsilon_N.
    \end{equation}  
    
    We bound $D\left(\rho^{\otimes N}\middle|\middle|\omega^{\prime}\right)$ from below and above.
    On the one hand, due to
    \begin{align}
        \lambda_{\min}(\omega^\prime)&\geq \frac{\epsilon_N}{(1+2\epsilon_N)}\lambda_{\min}^N,\\
        \lambda_{\min}(\sigma^\prime)&\geq \frac{2\epsilon_N}{(1+2\epsilon_N)}\lambda_{\min}^N\geq \frac{\epsilon_N}{(1+2\epsilon_N)}\lambda_{\min}^N,
    \end{align}
    by taking $\tilde{m} = \tfrac{\epsilon_N}{(1+2\epsilon_N)}\lambda_{\min}^N\in(0,1)$ in Corollary~\ref{cor:continuity_bound_relative_entropy_revised}, we obtain from~\eqref{eq:distance_omega_sigma}
    \begin{align}
        \label{eq:relative_entropy_continuity}
        \left|D\left(\rho^{\otimes N}\middle|\middle|\omega^\prime\right)-D\left(\rho^{\otimes N}\middle|\middle|\sigma^\prime\right)\right|\leq\frac{3\log_2^2\qty(\frac{(1+2\epsilon_N)}{\epsilon_N\lambda_{\min}^N})}{1-\frac{\epsilon_N\lambda_{\min}^N}{2(1+2\epsilon_N)}}\sqrt{\epsilon_N}=\frac{3\qty(\log_2\qty(\frac{1+2\epsilon_N}{\epsilon_N})+N\log_2\qty(\frac{1}{\lambda_{\min}}))^2}{1-\frac{\epsilon_N\lambda_{\min}^N}{2(1+2\epsilon_N)}}\sqrt{\epsilon_N}.
    \end{align}
    Thus, it holds that
    \begin{align}
        D\left(\rho^{\otimes N}\middle|\middle|\omega^\prime\right)&
        \stackrel{\eqref{eq:relative_entropy_continuity}}{\geq} D\left(\rho^{\otimes N}\middle|\middle|\sigma^\prime\right)-\frac{3\qty(\log_2\qty(\frac{1+2\epsilon_N}{\epsilon_N})+N\log_2\qty(\frac{1}{\lambda_{\min}}))^2}{1-\frac{\epsilon_N\lambda_{\min}^N}{2(1+2\epsilon_N)}}\sqrt{\epsilon_N}\\
        &\geq D\left(\rho^{\otimes N}\middle|\middle|\sigma_N\right)-\frac{3\qty(\log_2\qty(\frac{1+2\epsilon_N}{\epsilon_N})+N\log_2\qty(\frac{1}{\lambda_{\min}}))^2}{1-\frac{\epsilon_N\lambda_{\min}^N}{2(1+2\epsilon_N)}}\sqrt{\epsilon_N}\\
    \label{eq:lower_sigma_prime}
        &\geq \frac{1}{1+2\epsilon_N}D\left(\rho^{\otimes N}\middle|\middle|\sigma_N\right)-\frac{3\qty(\log_2\qty(\frac{1+2\epsilon_N}{\epsilon_N})+N\log_2\qty(\frac{1}{\lambda_{\min}}))^2}{1-\frac{\epsilon_N\lambda_{\min}^N}{2(1+2\epsilon_N)}}\sqrt{\epsilon_N}, 
    \end{align}
    where the second inequality holds due to $\sigma^{\prime} \in \mathcal{F}(\mathcal{H}^{\otimes N})$ and the fact that $\sigma_N$ is an optimal free state for the relative entropy of resource of $\rho^{\otimes N}$, 
    and the third inequality follows from $1 + 2\epsilon_N \geq 1$ and the non-negativity of the quantum relative entropy. 
    
    On the other hand, due to the joint convexity of the quantum relative entropy~\cite{N4}, we have
    \begin{align}
        D\left(\rho^{\otimes N}\middle|\middle|\omega^{\prime}\right)\leq\frac{1}{1+2\epsilon_N}D\left(\rho^{\otimes N}\middle|\middle|\tilde{\sigma}_N\right)+\frac{2\epsilon_N}{1+2\epsilon_N}D\left(\rho^{\otimes N}\middle|\middle|\frac{1}{2}\Delta+\frac{1}{2}\sigma_{\mathrm{full}}^{\otimes N}\right).
    \end{align}
    Using Lemma~\ref{lem:relative_entropy_upper_bound} with
    \begin{align}
        \lambda_{\min}\qty(\frac{1}{2}\Delta+\frac{1}{2}\sigma_{\mathrm{full}}^{\otimes N})&\geq\frac{\lambda_{\min}^N}{2},
    \end{align}
    we have
    \begin{equation}
        D\left(\rho^{\otimes N}\middle|\middle|\frac{1}{2}\Delta+\frac{1}{2}\sigma_{\mathrm{full}}^{\otimes N}\right)\leq N\log_2\qty(\frac{1}{\lambda_{\min}})+1.
    \end{equation}
    Therefore, we obtain
    \begin{equation}
    \label{eq:upper_sigma_prime}
        D\left(\rho^{\otimes N}\middle|\middle|\omega^{\prime}\right)\leq\frac{1}{1+2\epsilon_N}D\left(\rho^{\otimes N}\middle|\middle|\tilde{\sigma}_N\right)+\frac{2\epsilon_N\qty(N\log_2\qty(\frac{1}{\lambda_{\min}})+1)}{1+2\epsilon_N}.
    \end{equation}
    As a whole, from~\eqref{eq:lower_sigma_prime} and~\eqref{eq:upper_sigma_prime}, it follows that
    \begin{align}
        D\left(\rho^{\otimes N}\middle|\middle|\tilde{\sigma}_N\right)&\geq D\left(\rho^{\otimes N}\middle|\middle|\sigma_N\right)-\frac{3\qty(\log_2\qty(\frac{1+2\epsilon_N}{\epsilon_N})+N\log_2\qty(\frac{1}{\lambda_{\min}}))^2}{1-\frac{\epsilon_N\lambda_{\min}^N}{2(1+2\epsilon_N)}}\sqrt{\epsilon_N}\qty(1+2\epsilon_N)-2\epsilon_N\qty(
        N\log_2\qty(\frac{1}{\lambda_{\min}})+1
        )\\
    \label{eq:lower_sigma_N}
        &=D\left(\rho^{\otimes N}\middle|\middle|\sigma_N\right)-O\qty(\qty(\log_2\qty(\frac{1}{\epsilon_N})+N)^2\sqrt{\epsilon_N})\quad\text{as $N\to\infty$}.
    \end{align}

    Consequently, from~\eqref{eq:upper_sigma_N} and~\eqref{eq:lower_sigma_N}, we obtain
    \begin{align}
    \label{eq:sandwitch1}
        &\frac{D\left(\rho^{\otimes N}\middle|\middle|\sigma_N\right)
        }{N}-O\qty(\frac{\qty(\log_2\qty(\frac{1}{\epsilon_N})+N)^2\sqrt{\epsilon_N}}{N})
        \leq \frac{D\left(\rho^{\otimes N}\middle|\middle|\tilde{\sigma}_N\right)}{N}
        \leq \frac{D\left(\rho^{\otimes N}\middle|\middle|\sigma_N\right)}{N}.
    \end{align}
    Under the assumption of
    \begin{equation}
        \epsilon_N=o(1/N^2),
    \end{equation}
    it holds that
    \begin{equation}
        \lim_{N\to\infty}\frac{\qty(\log_2\qty(\frac{1}{\epsilon_N})+N)^2\sqrt{\epsilon_N}}{N}=0,
    \end{equation}
    where we use $\lim_{\epsilon\to +0}\epsilon^{c_1}\log^{c_2}\qty(1/\epsilon)=0$ for any $c_1,c_2>0$.
    Therefore, by taking the limit of $N\to\infty$ in~\eqref{eq:sandwitch1},
    we see that the limit
    \begin{equation}
        \lim_{N\to\infty}\min_{\tilde{\sigma}\in\mathcal{F}^{(\epsilon_N)}\qty(\mathcal{H}^{\otimes N})}\qty{\frac{D\left(\rho^{\otimes N}\middle|\middle|\tilde{\sigma}\right)}{N}}=\limsup_{N\to\infty}\min_{\tilde{\sigma}\in\mathcal{F}^{(\epsilon_N)}\qty(\mathcal{H}^{\otimes N})}\qty{\frac{D\left(\rho^{\otimes N}\middle|\middle|\tilde{\sigma}\right)}{N}}=\liminf_{N\to\infty}\min_{\tilde{\sigma}\in\mathcal{F}^{(\epsilon_N)}\qty(\mathcal{H}^{\otimes N})}\qty{\frac{D\left(\rho^{\otimes N}\middle|\middle|\tilde{\sigma}\right)}{N}}
    \end{equation}
    exists and coincides with
    \begin{equation}
        \lim_{N\to\infty}\min_{\tilde{\sigma}\in\mathcal{F}^{(\epsilon_N)}\qty(\mathcal{H}^{\otimes N})}\qty{\frac{D\left(\rho^{\otimes N}\middle|\middle|\tilde{\sigma}\right)}{N}}=\lim_{N\to\infty}\frac{D\left(\rho^{\otimes N}\middle|\middle|\tilde{\sigma}_N\right)}{N}=\lim_{N\to\infty}\frac{D\left(\rho^{\otimes N}\middle|\middle|\sigma_N\right)}{N}=\lim_{N\to\infty}\min_{\sigma\in\mathcal{F}\qty(\mathcal{H}^{\otimes N})}\qty{\frac{D\left(\rho^{\otimes N}\middle|\middle|\sigma\right)}{N}},
    \end{equation}
    which yields the conclusion.
\end{proof}

With the characterization of the quantum relative entropy shown in Proposition~\ref{prp:third_step} and the equivalence relation in Proposition~\ref{lem:limit_regularization}, 
we finally prove Theorem~\ref{thm:direct_part}. 
We recall the statement of the theorem below. 

\direct*

\begin{proof}
    For each $N\in\mathbb{N}$, we set $M(N)$ and $R(N)$ as any non-negative integers satisfying
\begin{align}
    &N-M(N)\geq2R(N),\\
    \label{eq:overall_NMR_1}
    &N-M(N)-R(N)\to\infty,\\
    \label{eq:overall_NMR_2}
    &\frac{M(N)R(N)}{N}\to\infty,\\
    \label{eq:overall_NMR_3}
    &\frac{M(N)}{N}\to0,\\
    \label{eq:overall_NMR_4}
    &\frac{R(N)}{N}\to0,
\end{align}
as $N\to\infty$.

Due to~\eqref{eq:minimization}, there exists a constant $\mu \in (0,1)$ such that 
\begin{equation}\liminf_{N\to\infty}\min_{\sigma\in\mathcal{F}\qty(\mathcal{H}^{\otimes N})}\qty{\Tr\qty[\qty(\rho^{\otimes N}-2^{yN}\sigma)_+]}=1-\mu \in (0,1). 
\end{equation}
Then, Proposition~\ref{prp:third_step} provides a sequence $\{\mu_N>0\}_{N\in\mathbb{N}}$ converging to
\begin{equation}
\label{eq:overall_mu_N}
    \lim_{N\to\infty}\mu_N=\mu,
\end{equation}
and it holds that
\begin{equation}
    \label{eq:overall_y}
    y\geq \limsup_{N\to\infty}\min_{\tilde{\sigma}\in\mathcal{F}^{(\epsilon_N)}\qty(\mathcal{H}^{\otimes N - M - R})}\qty{\frac{D\left(\rho^{\otimes N-M-R}\middle|\middle|\sigma\right)}{N - M - R}}, 
\end{equation}
where $\mathcal{F}^{(\epsilon_N)}\qty(\mathcal{H}^{\otimes N})$ is the set of asymptotically free states in Definition~\ref{def:asym_free_state} with $\{\epsilon_N > 0\}_{N\in\mathbb{N}}$ given by
\begin{equation}
\epsilon_N=\frac{2\mu_N^3}{2^{yN}}\qty(\frac{2\sqrt{2}}{\mu_N\mathrm{e}^{\frac{MR}{2N}}}+\frac{2\sqrt{2R}}{N}).
\end{equation}
Therefore, under our assumption of
\begin{equation}
    y > 0, 
\end{equation}
it follows from~\eqref{eq:overall_NMR_1},~\eqref{eq:overall_NMR_2},~\eqref{eq:overall_NMR_3},~\eqref{eq:overall_NMR_4}, and~\eqref{eq:overall_mu_N} that
\begin{equation}
    \epsilon_N=O\qty(1/2^{yN})=o(1/N^2) = o(1/(N-M-R)^2)\quad\text{as $N\to\infty$};
\end{equation}
thus, using Proposition~\ref{lem:limit_regularization}, we obtain from~\eqref{eq:overall_y}
\begin{align}
    \lim_{N\to\infty}\min_{\sigma\in\mathcal{F}(\mathcal{H}^{\otimes N})}\qty{\frac{D\left(\rho^{\otimes N}\middle|\middle|\sigma\right)}{N}}&=\lim_{N\to\infty}\min_{\sigma\in\mathcal{F}(\mathcal{H}^{\otimes N-M-R})}\qty{\frac{D\left(\rho^{\otimes N-M-R}\middle|\middle|\sigma\right)}{N-M-R}}\\
    \label{eq:relative_entropy_equivalence}
    &\stackrel{\text{Proposition~\ref{lem:limit_regularization}}}{=}\limsup_{N\to\infty}\min_{\tilde{\sigma}\in\mathcal{F}^{(\epsilon_N)}(\mathcal{H}^{\otimes N-M-R})}\qty{\frac{D\left(\rho^{\otimes N-M-R}\middle|\middle|\tilde{\sigma}\right)}{N-M-R}}\\
    &\stackrel{\eqref{eq:overall_y}}{\leq} y,
\end{align}
which yields the conclusion.
\end{proof}

\bibliographystyle{apsrmp4-2}
\bibliography{citation}
\end{document}